\documentclass{article}
\usepackage{tikz,hyperref,amsmath,amssymb,ntheorem,color, comment}
\usetikzlibrary{decorations.pathreplacing}
\usepackage[normalem]{ulem}
\usepackage{blkarray}
\usepackage{authblk}
\usetikzlibrary{decorations.pathreplacing}
\usetikzlibrary{automata} 
\usetikzlibrary{positioning} 
\usetikzlibrary{arrows} 
\usepackage{tikz}
\usetikzlibrary{automata} 
\usetikzlibrary{positioning} 
\usetikzlibrary{arrows.meta} 
\tikzset{node distance=1.5cm, 
            every state/.style={ 
                  semithick,
                  minimum size=7mm},
            initial text={}, 
            double distance=2pt, 
            every edge/.style={ 
                   draw,
                   ->,>=stealth', 
                   auto,
            semithick}}
\let\epsilon\varepsilon

\newtheorem{theorem}{Theorem}[section]
\newtheorem{lemma}[theorem]{Lemma}
\newtheorem{proposition}[theorem]{Proposition}
\newtheorem{corollary}[theorem]{Corollary}
{\theorembodyfont{\rmfamily}
  \newtheorem{example}[theorem]{Example}

  }
\newenvironment{proof}{\noindent\textit{Proof.}}{\QED\vskip\theorempostskipamount} 
\def\petitcarre{\vrule height4pt width 4pt depth0pt}
\def\QED{\relax\ifmmode\eqno{\hbox{\petitcarre}}\else{%
  \unskip\nobreak\hfil\penalty50\hskip2em\hbox{}\nobreak\hfil
  \petitcarre
  \parfillskip=0pt \finalhyphendemerits=0\par\smallskip}
  \fi}
\numberwithin{equation}{section}
\numberwithin{figure}{section}

\newcommand{\cA}{\mathcal{A}}
\newcommand{\cB}{\mathcal{B}}
\newcommand{\cC}{\mathcal{C}}
\newcommand{\cD}{\mathcal{D}}

\newcommand{\Z}{\mathbb{Z}}
\newcommand{\N}{\mathbb{N}}

\newcommand{\XS}{\mathsf{X}}
\DeclareMathOperator{\Card}{Card}

\usepackage{xspace}
\newcommand{\resp}{{resp.}\xspace }

\definecolor{lime}{HTML}{A6CE39}
\DeclareRobustCommand{\orcidicon}{%
	\begin{tikzpicture}
	\draw[lime, fill=lime] (0,0)
	circle [radius=0.16]
	node[white] {{\fontfamily{qag}\selectfont \tiny ID}};
	\draw[white, fill=white] (-0.0625,0.095)
	circle [radius=0.007];
	\end{tikzpicture}
	\hspace{-2mm}
}
\foreach \x in {A, ..., Z}{%
 \expandafter\xdef\csname orcid\x\endcsname{\noexpand%
 \href{https://orcid.org/\csname orcidauthor\x\endcsname}{\noexpand\orcidicon}}
}

\begin{document}
\title{One-sided Hom shifts}


\author[1]{Marie-Pierre B\'eal\orcidA{}}

\author[2]{Alexi Block Gorman}


\affil[1]{Univ. Gustave Eiffel, CNRS, LIGM, Marne-la-Vall\'ee, France}

\affil[2]{Institute for Logic, Language and Computation, Universiteit van Amsterdam}


\maketitle 
\begin{abstract}
We prove that it is decidable whether a one-sided shift of finite type is conjugate to a one-sided Hom-shift,
and whether a tree-shift of finite type is conjugate to a Hom tree-shift. The proof uses Williams's theory for one-sided shifts of finite type.
\end{abstract}

\section{Introduction} \label{section.introduction}
We consider two types of shift spaces, shifts of sequences and tree-shifts.
A shift of sequences is a set of one- or two-sided infinite sequences of symbols defined 
by a collection of finite words called forbidden blocks. A tree-shift is a set of infinite trees 
defined by a collection of forbidden blocks. 
A shift space of sequences is equipped with a shift transformation, whereas a tree-shift is 
equipped with multiple directional shift transformations.
A conjugacy between shift spaces is a continuous bijection that commutes with the shift transformations.
Shifts spaces defined by a finite set of forbidden blocks are called shifts of finite type.
Note that a one-sided shift of sequences can be seen as a tree-shift with arity $1$.

The problem of deciding whether two-sided finite-type shifts of sequences are conjugate remains open (see, for example, \cite{LindMarcus1995}). By contrast, the situation is significantly simpler for one-sided shifts: Williams \cite{Williams1973} proved that conjugacy is decidable for one-sided shifts of finite type. This decidability also extends to tree-shifts \cite{AubrunBeal2012}, a consequence of the inherently one-sided structure of trees.

Tree-shifts are, however, more complex than one-sided shifts. Petersen and Salama
\cite{PetersenSalama2018}, \cite{PetersenSalama2020} introduced the topological entropy for such spaces as an analog to topological entropy for traditional shift spaces to investigate their complexity. Their results were
generalized to asymptotic pressure for a broad class of systems in \cite{PetersenSalama2023}.
Properties of tree-shifts like entropy, transitivity, Hausdorff dimensions 
have been explored in \cite{BanCHW22}, \cite{BanEtal2025a}, \cite{BanEtal2025b}.

In this paper, we consider a particular class of shifts called Hom shifts.
A Hom shift is defined by an undirected, simple graph $G$ (where self-loops are allowed).
We also consider the class of directed Hom tree-shifts defined by a simple directed graph $G$ (again with self-loops). 
Directed Hom tree-shifts are called hom tree-SFT in \cite{BanEtAl2021}.
A Hom shift (\resp a directed Hom shift)
of sequences is the set of all sequences of symbols where the consecutive pair $ab$ is forbidden whenever 
$(a, b)$ is not an edge in the undirected (\resp directed) graph $G$. 

A Hom shift of trees (\resp a directed Hom tree-shift) is the set of all trees avoiding all blocks of height $2$ with a root labeled $a$ and having at least one child labeled $b$ such that $(a, b)$ is not an edge in the undirected (\resp directed) graph $G$. 

Alternatively, a Hom shift can be described as a nearest neighbour shift of finite type that is symmetric and isotropic. Hom shifts can be defined and studied in the multidimensional context (see 
\cite{CHANDGOTIA_2017},
\cite{ChandgotiaMarcus2018}, \cite{GangloffHellouinOpocha2024}). They form a natural subclass of shifts of finite type that arise in various contexts. For example, the hard square shift and the $n$-coloured chessboard shifts are both two-dimensional Hom shifts. 
Chandgotia and Marcus \cite{ChandgotiaMarcus2018} studied the mixing properties of Hom shifts and
related them to some questions in graph theory therein.

In this paper, we prove that it is decidable whether a one-sided shift of finite type is conjugate to a one-sided Hom shift, and whether a tree-shift of finite type is conjugate to a Hom tree-shift. 

The proof uses Williams's theory for one-sided shifts of finite type.
The total amalgamation of a directed multi-graph is the unique 
directed multi-graph obtained after a sequence of out-merging operations. 
In \cite{Williams1973}, Williams proved that if $G$ and $H$ are irreducible directed graphs that define irreducible one-sided edge shifts, then these shifts are conjugate if and only if $G$ and $H$ have the same total amalgamation.
We prove that a one-sided shift of finite type, or a tree-shift of finite type, is conjugate to a one-sided Hom shift
if and only if its total amalgamation satisfies some regularity conditions that are decidable.

For tree-shifts, it was proved in \cite{AubrunBeal2012} that two irreducible edge tree shifts
defined by a bottom-up tree automaton are conjugate if and only if they have the same total amalgamation (a bottom-up tree automaton obtained through merging operations). We give another proof of this result using trim top-down tree automata, and we
define the total amalgamation of a top-down edge tree automaton similarly. A top-down tree automaton is trim if there is at least one (top-down) transition starting at any state. The irreducibility is not required.
We prove that, if $\cA$ and $\cB$ are trim 
top-down edge tree automata defining two tree-shifts of finite type, then these shifts are conjugate if and only if $\cA$ and $\cB$ have the same total amalgamation. 

We prove that a tree-shift of finite type defined by a trim edge tree automaton is conjugate to a Hom tree-shift if the total amalgamation of this tree automaton satisfies some regularity conditions. 
The property is decidable in polynomial space and time. We also show that if two Hom tree shifts defined by undirected graphs $G$ and $H$, respectively, are conjugate, then $G$ is equal to $H$, up to a graph isomorphism.

We prove that a tree-shift of finite type defined by a trim edge tree automaton is conjugate to a directed Hom tree-shift if the total amalgamation of this tree automaton satisfies some symmetry conditions. 
The property is also decidable in polynomial space and time. 

The paper is organized as follows. We recall basic definitions of symbolic dynamics in Section \ref{section.symbolicdynamics}, including Williams's results on one-sided shifts of sequences of finite type. Basic definitions on tree shifts and the decidability of conjugacy for this class of shifts are presented in Section \ref{TreeShift}.
One-sided Hom shifts of sequences and Hom tree-shifts are studied in Section \ref{section.homshifts} and
\ref{section.hom.tree.shifts}. 
Directed Hom tree shifts are studied in Section  \ref{section.directed.hom.tree.shifts}.

\section{Symbolic dynamics}  \label{section.symbolicdynamics}
We briefly recall some basic definitions of symbolic dynamics. 
For a more complete presentation, see~\cite{LindMarcus1995}, \cite{Kitchens1998}, and \cite{BealBerstelEilersPerrin2010}.

\subsection{Words and sequences}
Let $A$ be a finite set, called the \emph{alphabet}. The elements of $A$ are
called \emph{letters}. A \emph{word} on $A$ is a finite sequence of elements of $A$.
The set of words on $A$ is denoted by $A^*$.

We consider the set $A^\Z$ of \emph{two-sided infinite sequences} 
of elements of $A$, equipped with the product of the discrete topology on $A$.


For $x=(x_n)_{n\in\Z}\in A^\Z$ with $x_n\in A$ and $i\le j$, we write 
$x_{[i,j)}=x_{i}x_{i+1}\cdots x_{j-1}$.
  A word $u$ \emph{occurs in} a sequence $x$  if $u=x_{[i,j)}$
  for some $i,j$. One also says that $u$ is a \emph{block}
  of $x$.

 Let $S$ denote the \emph{shift transformation} defined for $x\in A^\Z$ 
by $S(x)=y$ if $y_n=x_{n+1}$ for $n\in\Z$. It is continuous and one-to-one
from $A^\Z$ to itself.

 We also consider the corresponding sets $A^\N$ of right-infinite
sequences.
The \emph{one-sided shift transformation}
is defined similarly. For $x\in A^\N$, one has $y=S(x)$
if $y_n=x_{n+1}$ for $n\ge 0$. It is continuous, but not one-to-one (except when
$\Card(A)=1$).


\subsection{Shift spaces}
A set $X\subseteq A^\Z$  is a \emph{shift space} if it is topologically closed, and 
shift-invariant, that is, if $S(X)=X$. 

For a language $F \subseteq A^*$, 
the set of sequences $x\in A^\Z$ such that no element
of $F$ occurs in it is denoted by $\XS_F$. It is well known that 
 a set $X\subseteq A^\Z$ is a shift space if and only if
 $X=\XS_F$ for some $F\subseteq A^*$.
 A \emph{block} of $X$ is a word occurring in some sequence of $X$.

  A \emph{shift of finite type} is a shift
  $X=\XS_F$ for some finite set $F$. 

In the sequel, we consider labeled directed graphs (with labeled edges) and unlabeled directed graphs. We say graph for multigraph. 
A labeled or unlabeled directed graph $G$ is denoted by $G=(V, E)$, where $V$ is its set of vertices and $E$ its set of edges. In a labeled directed graph, an edge from a state $s$ to a state $t$ labeled by $a$ is denoted by $(s, a, t)$.

An \emph{edge shift} is the set of two-sided infinite paths in 
a finite directed graph.
An edge shift is a shift of finite type.

A shift space $X$ is \emph{irreducible} if and only if for any two blocks $u$, $v$ of $X$, there is a word $w$ such that $uwv$ is a block of $X$.

If $X$ is a shift space, we let $\cB_n(X)$ denote the set of blocks of $X$ of length $n$
and $\cB(X)$ denote the set of all blocks of $X$.

\subsection{One-sided shift spaces}


A \emph{one-sided shift space} is a closed 
subset $X$ of $A^\N$ such that $S(X) \subseteq X$. 
Note that one-sided shift spaces are usually defined as closed subsets such 
that $S(X) = X$, but we do not require this stronger condition here.
The set $A^\N$ itself is a one-sided shift
space, called the \emph{one-sided full shift}.

For a two-sided sequence $x\in A^\Z$, we define $x^+=x_0x_1\cdots$.
If $X$ is a two-sided shift space, then the set $X^+ = \{x^+ \mid x \in X\}$
is a one-sided shift space.

A one-sided shift space is \emph{of finite type} if it is the set $\XS_F$ of one-sided sequences over $A$ avoiding all words of some finite set $F \subseteq A^*$.
A \emph{one-sided edge shift} is the set $X_G$ of right-infinite paths in 
a finite directed graph $G$. Note that the paths may start at any state.

The \emph{adjacency matrix} of a labeled or unlabeled directed graph $G$ with a set of vertices $E$ is the matrix $M=(M_{st})_{s,t \in E}$ such that $M_{st}$ is the number of edges going from $s$ to $t$ in $G$. Thus, an edge shift is also defined by a square matrix with coefficients in~$\N$.


\begin{example} \label{exampleEdgeShift}
   The one-sided edge shift $X_G$ represented by the directed graph $G$ of Figure~\ref{figureEdgeShift}:

\begin{figure}[h]
\centering
    \begin{tikzpicture}[shorten >=1pt,node distance=2.5cm, every edge/.style={draw,->,>=stealth',auto,semithick},auto]
      \node[state] (1) {$1$};
      \node[state, right of=1] (2) {$2$};
      \draw (1) edge[bend left=30] node {} (2);
      \draw (1) edge[loop above] node {} (1);
      \draw (2) edge[bend left=30] node {} (1);
    \end{tikzpicture}
      \caption{A one-sided edge shift.}\label{figureEdgeShift}
\end{figure}
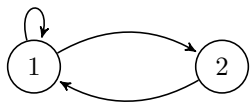
is also defined by the adjacency matrix of $G$, that is, by the matrix 
  \[
  M =
  \begin{bmatrix}
  1 & 1 \\
  1 & 0
  \end{bmatrix}.
  \]
\end{example}

\subsection{One-sided conjugacy}

A \emph{(one-sided) morphism} $\varphi \colon X \subseteq A^\N \to B^\N$, where $X$ is a one-sided shift space, 
is a continuous map commuting with the shift transformation: $S(\varphi(x)) = \varphi(S(x))$, for each $x \in X$.

Let $X$ be a one-sided shift space over the alphabet $A$, and let $B$ be another alphabet.
Given a nonnegative integer $n$,
a \emph{block map}
is a map $f\colon \cB_{n+1}(X)\to B$.
The \emph{(one-sided)-sliding block code}
defined by $f$ is the map $\varphi:X\to B^\N$ defined by
$\varphi(x)=y$ if for every $i\in\N$, $y_i=f(x_{[i,i+n]})$.
\begin{figure}[htbp]
\centering

\tikzset{node/.style={draw,minimum size=0.4cm,inner sep=0pt}}
	\tikzset{title/.style={minimum size=0.5cm,inner sep=0pt}}
        \begin{tikzpicture}
      \node[title]at(-4,1){};
     \node[title]at(-2.8,1){};
  \node[title](xi-m)at(-1.5,1){$\ x_{0} \cdots$};
  \node[title,text width=.6cm]at(-.7,1){$x_{i-1}$};
  \node[node](x0)at(0,1){$x_i$};
  \node[node,text width=1cm]at(.7,1){$\ \ \cdots$};
  \node[node,text width=1cm](xi+n)at(1.7,1){$\ x_{i+n}$};
  \node[title]at(2.8,1){$x_{i+n+1}$};
  \node[title]at(4,1){$\cdots$};
  \node[title]at(-1.5,0){$y_0 \cdots $};
  \node[title,text width=.6cm]at(-.7,0){$y_{i-1}$};
  \node[node]at(0,0){$y_i$};
  \node[title,text width=.6cm]at(.5,0){$\ y_{i+1}$};
   \node[title]at(1.5,0){$\cdots$};

  \draw[->,left](0,.8)--node{$f$}(0,.2);
\end{tikzpicture}
\caption{A one-sided sliding block code.}\label{figureSlidingBlockCode}
\end{figure}

Thus, $y$ is computed from $x$ by sliding a window of length $n+1$
on $x$ as in Figure~\ref{figureSlidingBlockCode}.
The integer $n$ is the \emph{anticipation} of $\varphi$
and $f$ is said to be an $n$-block map.

The set $\varphi(X)$ is a one-sided shift space.

It is a classical result that $\varphi$ is a morphism if and only if it is a (one-sided) sliding block code. 
A \emph{(one-sided) conjugacy} is an invertible (one-sided) morphism. Its inverse is also a (one-sided) sliding block code, perhaps defined with a block map having a different anticipation. If $\varphi \colon X \subseteq A^N \to B^N$, where $X$ is a one-sided shift space, is a conjugacy,  we say that $X$ and $Y = \varphi(X)$ are \emph{conjugate}. Any one-sided shift of finite type is conjugate to a one-sided edge shift (see for instance \cite{LindMarcus1995}).

\begin{example} \label{example.conjugate}
The one-sided shift of labels of right-infinite paths of the labeled directed graph in Figure \ref{figure.conjugate} is a shift of finite type conjugate to the edge shift of Example \ref{exampleEdgeShift}.
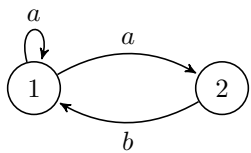
\begin{figure}[h]
\centering
    \begin{tikzpicture}[shorten >=1pt,node distance=2.5cm, every edge/.style={draw,->,>=stealth',auto,semithick},auto]
      \node[state] (1) {$1$};
      \node[state, right of=1] (2) {$2$};
      \draw (1) edge[bend left=30] node {$a$} (2);
      \draw (1) edge[loop above] node {$a$} (1);
      \draw (2) edge[bend left=30] node {$b$} (1);
    \end{tikzpicture}
      \caption{A one-sided shift of finite type.}\label{figure.conjugate}
\end{figure}
\end{example}
\subsection{Out-splitting and out-merging}
In this section, we recall the notion of transformations on directed graphs called out-splittings and out-mergings. The transformation can be defined either for labeled directed graphs or for directed graphs.

Let $X_G$ be a one-sided edge shift defined by a directed graph $G=(V, E)$. We may assume that the graph is \emph{trim}, that is, that
each vertex has at least one outgoing edge.
An \emph{out-splitting} of $G$ is a local transformation of $G$ into a graph $G'=(V', E')$ obtained by selecting a vertex $s$ and partitioning the set of edges going out of $s$ into two non-empty sets $E_1$ and $E_2$.
The graph $G'$ is defined as follows: 

\begin{itemize}
\item $V' = V \setminus \{s\} \cup \{s_1, s_2\}$,
\item $E'$ contains all edges of $E$ neither starting at or ending in $s$,
\item $E'$ contains the edge $(s_1, a, t)$ for each edge $(s, a, t) \in E_1$, and the edge $(s_2, a, t)$ for each edge $(s, a, t) \in E_2$, so long as $t \neq s$,
\item $E'$ contains the edges $(t, a, s_1)$ and $(t, a, s_2)$ if $(t, a, s)$  in $E$,
when $t \neq s$,
\item $E'$ contains the edges $(s_1, a, s_1)$ and $(s_1, a, s_2)$ 
if $(s, a, s)$ in $E_1$, and the edges $(s_2, a, s_1)$ and $(s_2, a, s_2)$
if $(s, a, s) \in E_2$.
\end{itemize}

Note that for each self-loop from $s$ to $s$ in $E$, if this loop is in $E_1$, then $E'$ has a self-loop from $s_1$ to $s_1$ and an edge from $s_1$ to $s_2$. Informally, an out-splitting operation involves splitting a vertex in $G$ into two distinct vertices, partitioning its outgoing edges between them, while preserving the incoming edges by duplicating them for each new vertex (see Figure~\ref{figureSplitting}). 

\begin{example}

The graph $G'$ in the right part of Figure~\ref{figureSplitting} is an out-split of the graph $G$ in the left part of the figure.
Here, $s = 1$, and the partition of the outgoing edges of $1$ is $\{E_1, E_2\}$, where $E_1$ contains the loop around $1$, and $E_2$ contains the two edges going from $1$ to $2$.

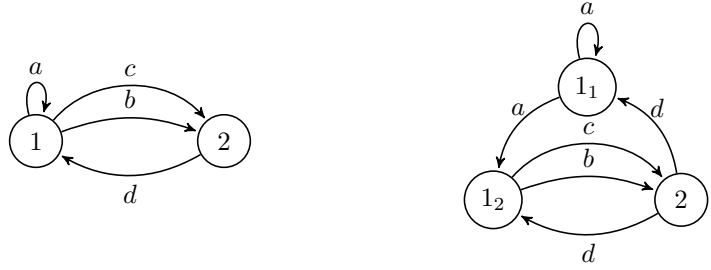
\begin{figure}[htbp]
\begin{minipage}{0.5\textwidth}
\centering
    \begin{tikzpicture}[shorten >=1pt,node distance=2.5cm, every edge/.style={draw,->,>=stealth',auto,semithick},auto]
      \node[state] (1) {$1$};
      \node[state, right of=1] (2) {$2$};
      \draw (1) edge[bend left=20] node {$b$} (2);
      \draw (1) edge[bend left=50] node {$c$} (2);
      \draw (1) edge[loop above] node {$a$} (1);
      \draw (2) edge[bend left] node {$d$} (1);
    \end{tikzpicture}
  \end{minipage}%
  \begin{minipage}{0.5\textwidth}
   \centering
    \begin{tikzpicture}
    [shorten >=1pt,node distance=2.5cm, every edge/.style={draw,->,>=stealth',auto,semithick},auto]
      \node[state] (1) at (1.25,0) {$1_1$};
      \node[state] (2) at (2.5,-1.5) {$2$};
      \node[state] (3) at (0,-1.5) {$1_2$};
     \draw (1) edge[loop above] node {$a$} (1);
      \draw (1) edge[bend right, above]  node {$a$} (3);
      \draw (3) edge[bend left=20] node {$b$} (2);
      \draw (3) edge[bend left=50] node {$c$} (2);
      \draw (2) edge[bend right, above] node {$d$} (1);
      \draw (2) edge[bend left] node {$d$} (3);
    \end{tikzpicture}
      \end{minipage}%
 \caption{An out-splitting.}
 \label{figureSplitting}
\end{figure}
\end{example}

Let $X_G$ be the one-sided edge shift defined by $G$ and $X_{G'}$ be the one-sided edge shift defined by $G'$. Then $X_G$ and $X_{G'}$ are conjugate. Indeed, let $\varphi \colon: E^\N \to E'^\N$
be the sliding block code defined by the $2$-block map $f\colon \cB_2(X_G) \to E'$, where 

\begin{align*}
    f((t, a, u)(u, b, v)) &= (t, a, u)   && \text{if } t, u \neq s, \\
    f((t, a, s)(s, b, v)) &= (t, a, s_1) && \text{if } t \neq s \text{ and } (s, b, v) \in E_1,\\
    f((t, a, s)(s, b, v)) &= (t, a, s_2) && \text{if } t \neq s \text{ and } (s, b, v) \in E_2,\\
    f((s, a, t)(t, b, u)) &= (s_1, a, t) && \text{if } t \neq s \text{ and } (s, a, t) \in E_1,\\
    f((s, a, t)(t, b, u)) &= (s_2, a, t) && \text{if } t \neq s \text{ and } (s, a, t) \in E_2,\\
    f((s, a, s)(s, b, t)) &= (s_1, a, s_1) && \text{if } (s, b, t) \in E_1,\\
    f((s, a, s)(s, b, t)) &= (s_1, a, s_2) && \text{if } (s, b, t) \in E_2.
\end{align*}
defines a conjugacy from $X_G$ onto $X_{G'}$. 
Its inverse is the sliding block code defined by the $1$-block map $g\colon  E' \to E$, where $g(t, \ell , u) = (\pi(t), \pi(\ell),\pi(u))$, with $\pi(t) = t$ if $t \neq s_1, s_2$, $\pi(s_1) = \pi(s_2) = s$, $\pi(a) = a$.

Following \cite{LindMarcus1995}, we define the notion of \emph{general out-splitting}: the outgoing edges from a given state can be partitioned into arbitrary many subsets instead of just two, and the partitioning can occur simultaneously at all of the states instead of just one.  This procedure also accommodates self-loops. 

The inverse operation of an out-splitting is referred to as an \emph{out-merging}.
An out-merging of a directed graph $G'= (V', E')$ can be performed 
if there are two vertices $s_1,s_2$ of $G'$ such that the adjacency matrix $M'$ satisfies:
\begin{itemize}
\item  the column of index $s_1$ is equal to the column of index $s_2$ of $M'$.
\end{itemize}
The adjacency matrix of $G$ is thus the matrix $M$ obtained by adding the rows of index $s_2$ to the row of index $s_1$ of $M'$
and then removing the column of index $s_2$ afterward. 
The graph $G$ is called an \emph{elementary amalgamation} of $G'$. Notice that even if $M'$ has \( 0\text{-}1 \) entries, $M$ may not have \( 0\text{-}1 \) entries.

Let $M'$ be the adjacency matrix of a directed graph $G'$, and $(V_1, V_2, \ldots, V_k)$ be a partition of $V'$ into classes
such that if $s, t$ belong to the same class, then the columns of indices $s$ and $t$ of $M'$ are identical.
When at least one set of the partition has a size greater than $1$, we can perform a \emph{general merging}. 
We define a graph $K$ of adjacency matrix $N$ obtained by merging all states of each $V_i = \{s_{i,1}, \ldots s_{i,k_i}\}$ into a single state
$s_{i,1}$. 

The row in $N$ corresponding to $s_{i,1}$ is obtained by summing the rows of the states of $V_i$ in $M'$ and removing the columns $s_{i,2}, \cdots, s_{i,k_i}$.
The graph $K$ is called a \emph{general amalgamation} of $G'$.

A directed graph $G$ that can be obtained from $G'$ by a sequence of elementary
out-mergings is called an \emph{amalgamation} of $G$.

The following proposition, due to R. F. Williams \cite{Williams1973}, shows that two out-merging transformations commute (see also \cite{BoyleFranksKitchens1990} for a proof).  

\begin{proposition}[Williams \cite{Williams1973}]  \label{propositionWilliams1}
If $G$ and $H$ are amalgamations of a common directed graph $L$, then they have a common amalgamation $K$.
\end{proposition}

As a consequence, given a directed graph $G$, there is a unique graph, up to a renaming of the vertices, 
obtained by performing elementary out-mergings until we cannot perform anymore.
This graph is called the \emph{total amalgamation} of $G$.


\begin{proposition}[Williams \cite{Williams1973}] \label{propositionWilliams2}
Let $G$ and $H$ be irreducible directed graphs that define one-sided edge shifts $X_G$ and $X_H$. Then $X_G$ and $X_H$ are conjugate if and only if $G$ and $H$ have the same total
amalgamation.
\end{proposition}

Proposition \ref{propositionWilliams2} also holds for one-sided edge shifts defined by 
trim directed graphs.
We will demonstrate this fact in section \ref{TreeShift}, as it follows directly from the propositions therein.

\section{Hom shifts}  \label{section.homshifts}
In this section, we consider shift spaces defined by undirected graphs, which we assume to be simple, that is, with at most one (undirected) edge between two vertices, and where loops are allowed. 
We will consider
shift spaces defined by simple directed graphs in Section \ref{section.directed.hom.tree.shifts}.

\subsection{Hom shifts of sequences} \label{d=1}
A \emph{Hom shift} is defined by a simple (unlabeled) undirected graph $G=(V, E)$ as the set $X(G)$ of bi-infinite sequences $x=(x_i)_{i \in \Z}$
of states of $G$ such that $(x_i, x_{i+1})$ is an edge of $G$. In this way, each sequence $x$ can be viewed as a graph homomorphism
$x \colon \Z \to G$. The shift $X(G)$ can be seen as the set of graph homomorphisms, which explains the name `Hom shift', coined by Chandgotia in \cite{CHANDGOTIA_2017}.
A \emph{one-sided Hom shift} is a sequence $x=(x_i)_{i \in \N}$
of states in $G$ such that $(x_i, x_{i+1})$ is an edge in $G$. It is also denoted $X(G)$.

To a simple undirected graph $G=(V,E)$ corresponds its \emph{underlying directed graph} $G'=(V',E')$ where $V'=V$ and $(s,t) \in E'$, $(t,s) \in E'$ whenever $(s,t) \in E$.
A one-sided Hom shift is a one-sided shift of finite type. Indeed, it is equal to $X_F$, where $F$ is the finite set $\{st \mid (s,t) \notin E\}$.

\begin{lemma} \label{lemma.homshifts}
A one-sided Hom shift defined by an undirected graph is conjugate to the one-sided edge shift
defined by its underlying directed graph.
\end{lemma}
\begin{proof}
    Let $G = (V, E)$ be an undirected graph and $G'= (V, E')$ its underlying directed graph.
    Let $\varphi$ be the
sliding block code defined by the $1$-block map $f \colon E' \to V$ with $f((s,t)) = s$.
The map $\varphi$
 is a conjugacy from $X_{G'}$ to $X(G)$. Its inverse is the sliding block code defined by the $2$-block map $g \colon V^2 \to E'$ with 
 $g(st) = (s, t) $.
\end{proof}

\begin{example} \label{exampleHomShift}
A one-sided Hom shift $X$ on the alphabet $A=\{a, b\}$ is defined by the undirected graph in the left part of Figure~\ref{figureHomShift}.
It is the set of right-infinite sequences avoiding the block $bb$.
The underlying directed graph $G'=(V, E')$ is represented in the right part of Figure~\ref{figureHomShift}. 
It defines a one-sided edge shift $X_{G'}$ conjugate to $X(G)$.

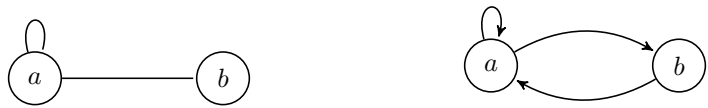
\begin{figure}[h]
\begin{minipage}{0.5\textwidth}
\centering
    \begin{tikzpicture}  
    [shorten >=1pt,node distance=2.5cm, every edge/.style={draw, semithick},auto]
      \node[state] (1) {$a$};
      \node[state, right of=1] (2) {$b$};
      \draw (1) edge node {} (2);
      \draw  (1)  edge[loop above, every loop/.style={-, semithick, draw}]  node  {} (1);
    \end{tikzpicture}\\
  \end{minipage}%
  \begin{minipage}{0.5\textwidth}
   \centering
    \begin{tikzpicture}[shorten >=1pt,node distance=2.5cm, every edge/.style={draw,->,>=stealth',auto,semithick},auto]
      \node[state] (1) {$a$};
      \node[state, right of=1] (2) {$b$};
      \draw (1) edge[bend left]  node {} (2);
      \draw (2) edge [bend left] node {} (1);
      \draw (1)  edge[loop above] node  {} (1);
    \end{tikzpicture}\\
      \end{minipage}%
 \caption{The Hom shift $X(G)$ on the left, and the edge shift $X_{G'}$ on the right.}\label{figureHomShift}
\end{figure}
\end{example}
Notice that $X_{G'}$ itself is not a Hom shift. Indeed, $(a, a)(a, b)$ is a block of $X_{G'}$ while $(a, b)(a, a)$ is not.
Thus, the Hom shift property is not invariant under conjugacy.


We say that a directed graph $G$ is \emph{regular} if 
\begin{itemize}
\item for each vertex $s$ of $G$, there is a positive integer $n_s$ such that for each vertex $t$ of $G$ the number of edges going from $s$ to $t$ is either $0$ or $n_s$. Equivalently, for each state $s$, the nonzero values of the row of index $s$ in the adjacency matrix of $G$ are identical.
\item whenever there is at least one edge from $s$ to $t$ in $G$, then there is at least one edge from $t$ to $s$ in $G$.
\end{itemize}

Note that the underlying directed graph of a Hom shift is regular. Indeed, the two conditions follow directly from the definitions since the adjacency matrix is a symmetrical $0$-$1$-matrix.

\begin{proposition} \label{propositionRegular1}
 Any amalgamation of a regular directed graph is a regular directed graph.
\end{proposition}
\begin{proof}
Without loss of generality, we may consider only an elementary amalgamation. 
We show that the definition of an elementary amalgamation implies that the two conditions of regularity are maintained. 

Let $s_1$ and $s_2$ be two vertices of a regular directed graph $G = (V, E)$ with adjacency matrix $M$. Let $N$ be the adjacency matrix of the amalgamation $G'$ after merging $s_1$ and $s_2$ into a state $s_{12}$. 
Let $n_1$ (\resp $n_2$) be the value of the nonzero entries of the row of $M$ of index $s_1$ (\resp $s_2$). 

Suppose that $s_1$ and $s_2$ satisfy the prerequisites for an out-merging,
i.e.\ the columns corresponding to $s_1$ and $s_2$ in $M$ are identical.
For $s \in V$, we define $\pi(s)$ as $s$ if $s \neq s_1, s_2$, and 
$\pi(s_1) = \pi(s_2) = s_{12}$.

Let us now show that the second condition of regularity is verified. 
If $N_{\pi(s), \pi(t)} > 0$, then there exist $s'$ with $\pi(s') = \pi(s)$ such that for each $t'$ with $\pi(t') = \pi(t)$, one has $M_{s', t'} > 0$.
This implies that $M_{t', s'} > 0$ by regularity of $G$, and thus 
$N_{\pi(t), \pi(s)} > 0$.

We now show that the first condition of regularity is verified. 
If $N_{\pi(s), \pi(t)} > 0$, then there exist $s'$ with $\pi(s') = \pi(s)$ such that for each $t'$ with $\pi(t') = \pi(t)$, one has $M_{s', t'} > 0$. Hence, $M_{s', t'}  = n_{s'}$.
By construction, if $s \neq s_1, s_2$, we have $s' = s$ and we get $N_{\pi(s'), \pi(t')} = N_{\pi(s), \pi(t)} = n_{s'} = n_s$. 
If $s = s_1$ or $s = s_2$, then $s' = s_1$ or $s'= s_2$.
Since $M_{s', t'} > 0$,  $M_{t', s'} > 0$, implying $M_{t', s_1} > 0$ and $M_{t', s_2} > 0$,
and $M_{s_1, t'} > 0$ and $M_{s_2, t'} > 0$.
Thus, $M_{s_1, t'} = n_1$ and $M_{s_2, t'} = n_2$, implying $N_{s_{12}, \pi(t)} = n_1 + n_2$. 
Thus, $G'$ is regular.

\end{proof}

Note that the graph obtained after an out-splitting of a regular graph is not necessarily regular, as one can see in Figure \ref{figureSplitting}.
Indeed, there is an edge from $1_1$ to $1_2$ and there is no edge from $1_2$ to $1_1$
in the graph in the right part of the figure.

\begin{proposition}  \label{propositionRegular2}
Any one-sided edge shift defined by a regular directed graph is conjugate to a one-sided Hom shift.
\end{proposition}

\begin{proof} 
Let $G$ be a regular directed graph with adjacency matrix $M$, and $X_G$ be the one-sided edge shift defined by $G$. We construct a new directed graph $G'= (V', E')$ with adjacency matrix $M'$ as follows.
For each vertex $s$ of $G$ whose corresponding row in $M$ has all its nonzero entries equal to $n_s$, we define $n_s$ vertices $s_1, \cdots s_{n_s}$ in the graph $G'$,
and add exactly one edge from each $s_i$ to each $t_j$ when $M_{s, t} = n_s$.

By construction, the adjacency matrix $M'$ of $G'$ is a $0\text{-}1$-matrix. Since $G$ is regular,  $M_{s, t} > 0$ if and only if $M_{t, s} > 0$. By construction, $M'_{s_i,t_j} = 1$
if and only if $M_{s, t} > 0$. Thus, $M'_{s_i,t_j} = 1$ if and only if $M'_{t_j, s_i} = 1$,
The graph $G$ is a general amalgamation of $G'$. Hence, the one-sided edge shift $X_G$ is conjugate to the one-sided edge shift $X_{G'}$. 

Let $H$ be the simple undirected graph whose set of vertices is $V'$, and where $(s_i, t_j)$ is an (undirected) edge of $H$ if and only if $M'_{s_i, t_j} = 1$. 
The one-sided edge shift $X_{G'}$ is the one-sided edge shift defined by the underlying directed graph of $H$.
By Lemma \ref{lemma.homshifts}, the edge shift $X_{G'}$ is conjugate to the Hom shift $X(H)$. 
Thus, $X_{G}$ is conjugate to a Hom shift as a one-sided shift.
\end{proof}
We now give an effective characterization of one-sided edge shifts conjugate to a one-sided Hom shift.
In the sequel, we will assume that the directed graphs are trim.

\begin{proposition} \label{proposition.amalgamation1}
A one-sided edge shift defined by a trim directed graph $G$ is conjugate to a Hom shift if and only if the total amalgamation of $G$ is regular.
\end{proposition}
\begin{proof}
Let $X_G$ be the one-sided edge shift defined by a trim directed graph $G$. Assume that $X_G$ is conjugate to a Hom shift $X(H)$, where $H$ is a simple undirected graph.

The one-sided shift $X_G$ is thus conjugate to the edge shift $X_{H'}$ defined by the underlying directed graph $H'$ of $H$, which is regular.
As a consequence of Proposition \ref{propositionWilliams2} and since $G$ and $H'$ are trim, the total amalgamations of $G$ and $H'$ are the same, up to a renaming of the vertices.
By Proposition \ref{propositionRegular1}, this total amalgamation is regular.

Conversely, if $G$ has a total amalgamation $K$ that is regular, then the one-sided shift $X_G$ is conjugate 
to the one-sided edge shift $X_K$, which is conjugate to a Hom shift by Proposition \ref{propositionRegular2}. 
\end{proof}

\begin{corollary}\label{corollary.decidable.homshift}  
It is decidable whether a one-sided shift of sequences of finite type is conjugate to a one-sided Hom shift.
\end{corollary}
\begin{proof}
Let $X = \XS_F$ be a one-sided shift of finite type over $A$. We may assume that all words of $F$ have the same length $n$. Then $\XS_F$ is conjugate to $X_G$, where $G$ is the trim part of the directed graph whose set of states is $A^{n-1}$, and the edges $(au, ub)$ with that $a, b \in A$ and $aub \notin F$.
Thus, we may assume that $X$ is the one-sided edge shift defined by a trim directed graph $G$. 
We apply the construction of Proposition \ref{propositionWilliams2} to build the total amalgamation of 
$G$.  We conclude with Proposition \ref{propositionRegular1} by checking whether this total amalgamation is regular.

If the input is a trim directed graph $G$ that defines the one-sided edge shift,
the time and space complexity are polynomial. Indeed, if $G$ has $n$ vertices, the total amalgamation can be performed in time $O(n^2)$ with a linear lexicographical sort of the columns of the adjacency matrix of $G$ for each general amalgamation.

\end{proof}

\begin{example}
Let $X_G$ be the one-sided edge shift defined by the graph $G$ with the following adjacency matrix $M$:
\[
M = \begin{bmatrix}
2 & 2 & 1\\
1 & 1 & 2 \\
1 & 1 & 0
\end{bmatrix}.
\]
The first and the second columns corresponding to states $1$ and $2$ are identical. Thus, states $1$ and $2$ can be merged
and the amalgamation is the graph $K$ with the following adjacency matrix $N$:
\[
N = \begin{bmatrix}
3 & 3 \\
1 & 0  
\end{bmatrix}.
\]
After this out-merging, no other out-merging is possible. Hence,
the total amalgamation of $G$ is the graph $K$.
Since $K$ is regular, $X_G$ is conjugate to a one-sided Hom shift, the shift $X(H)$ defined by the undirected graph of Figure~\ref{figure.amalgamation}.
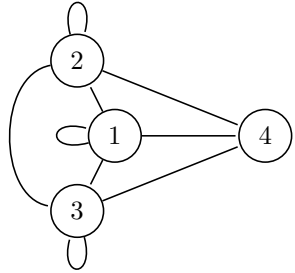
\begin{figure}[h]
\centering
    \begin{tikzpicture}  
    [shorten >=1pt,node distance=2.5cm, 
    every edge/.style={draw, semithick},auto]
      \node[state] (1) at (0.5,0) {$1$};
      \node[state] (2) at (0,1) {$2$};
      \node[state] (3) at (0,-1) {$3$};
      \node[state] (4) at (2.5,0) {$4$};
      \draw (1) edge node {} (2);
      \draw (1) edge node {} (3);
      \draw (2) edge[bend right=80] node {} (3);
       \draw (1) edge node {} (4);
       \draw (2) edge node {} (4);
       \draw (3) edge node {} (4);
      \draw  (1)  edge[loop left, every loop/.style={-, semithick, draw}]  node  {} (1);
      \draw  (2)  edge[loop above, every loop/.style={-, semithick, draw}]  node  {} (2);
      \draw  (3)  edge[loop below, every loop/.style={-, semithick, draw}]  node  {} (3);
    \end{tikzpicture}\\
  \caption{The one-sided Hom shift $X(H)$ conjugate to $X_G$.} \label{figure.amalgamation}
\end{figure}
\end{example}

The following proposition shows that two one-sided Hom shifts defined by distinct simple undirected graphs are never conjugate.

\begin{proposition}  \label{proposition.amalgamation2}
Let $X(G)$ and $X(H)$ be two one-sided Hom shifts defined by the undirected graphs $G$ and $H$, respectively.
If $X(G)$ and $X(H)$ are conjugate, then the undirected graphs $G$ and $H$ are isomorphic as undirected graphs.
\end{proposition}

\begin{proof} 
Let $G'$ and $H'$ be the underlying directed graphs of $G$ and $H$, respectively. The graphs $G'$ and $H'$ are regular and trim by construction. By Lemma \ref{lemma.homshifts}, $X(G)$ (resp. $X(H)$) is conjugate to $X_{G'}$ (resp. $X_{H'}$). Let $M$ (resp. $M'$) be the adjacency matrix of $G'$ (reps. $H'$). The matrices $M$ and $M'$ have $0\text{-}1$ entries.
By Proposition \ref{proposition.amalgamation1}, if $X(G)$ and $X(H)$ are conjugate, then 
there is a regular directed graph $K$ such that
$K$ is the total amalgamation of $G'$ and is the total amalgamation of $H'$.

We define the partition $(V_1, \ldots, V_\ell)$ of the vertices of $G'$ where the sets $V_i$
are the equivalence classes of the equivalence relation $t \sim u$ if and only if the columns of $M$ corresponding to $t$ and $u$ are identical.
Let us compute an initial general amalgamation $K'$ of $G'$ defined by this partition.
All states of $V_i$ are merged into a single state $s_{i}$ in~$K'$.
Since $G'$ is regular and $M$ has $0\text{-}1$ entries, if $M_{s, t} = 1$ for some $s \in V_i,  t \in V_j$, then $M_{s, t'} = 1$ for all $t' \in V_j$, and $M_{t, s'} = 1$
for all $s'\in V_i$. 
Thus, whenever $M_{s, t} = 1$ for some $s \in V_i, t \in V_j$, then $M_{s', t'} = 1$ for all $s' \in V_i, t' \in V_j$.

Hence, the adjacency matrix \( N \) of the general amalgamation \( K' \) satisfies either \( N_{s_k, s_i} = 0 \) or \( N_{s_k, s_i} = n_k \), where \( n_k = |V_k| \).

Assume that the two (distinct) states $s_i$ and $s_j$ of $K'$ can be merged through another out-merging after this first round. This implies that $N_{s_k, s_i} = N_{s_k, s_j}$ for each state $s_k$ of $K'$.
Hence, $M_{s, t} = M_{s, u}$ for each state $s$ of $G'$, each state $t$ of $V_i$, and each state $u$ of $V_j$. This implies $s_i = s_j$ by definition of the partition $(V_1, \ldots, V_\ell)$.   
As a consequence, no other merging can be performed, and $K'$ is the total amalgamation $K$ of $G'$. 

Similarly, if $(W_1, \ldots, W_{\ell'})$ is the partition of the vertices of $H'$ corresponding to the equivalence classes of the equivalence relation $t \sim u$ if and only the columns of $M'$ corresponding to $t$ and $u$ are identical.
The total amalgamation $K$ of $H'$ is obtained by merging all states in each $W_i$.
After renumbering the sets $W_i$, we may assume that all states of $W_i$ are merged
into $s_i$ in $K$. 

Hence, $\ell = \ell'$. Further, $N_{s_k, s_i} = |V_k|$ and $N_{s_k, s_i} = |W_k|$, implying
$|V_k|= |W_k|$ for each $k$. After renaming the states, we may thus assume that $V_k = W_k$ for each~$k$. As a consequence, $M = M'$, implying $G' = H'$.
\end{proof}

\section{Tree-shifts }\label{TreeShift}
Let $\Sigma=\{0,1,\ldots, d-1\}$ be a finite alphabet of cardinality $d$.
An \emph{(infinite) tree}~$t$, with nodes labeled on a finite alphabet $A$, is a total
function from $\Sigma^*$ to~$A$.  A node is a word of $\Sigma^*$. The empty
word, denoted by $\varepsilon$, corresponds to the root of the tree. 
If $x$ is a node, its children are $xi$ with $i \in
\Sigma$. 
If $t$ is a tree and $x$ is a node, $t(x)$ is sometimes
denoted by $t_x$. 
A \emph{(infinite) path} in a tree $t$ is a sequence
$(t_{x_n})_{n\geq 0}$ where $x_{n+1} \in 
x_n\Sigma$ for any $n \geq 0$.   
We let $T(A)$ denote the set of all infinite trees labeled on $A$.

For each $i \in \Sigma$, we define the shift transformation $\sigma_i \colon T(A) \to T(A)$
as follows. 
If $t$ is a tree, $\sigma_i(t)$ is the
tree rooted at the $i$-th child of $t$, i.e. $\sigma_i(t)_x=
t_{ix}$ for any $x \in \Sigma^*$. 
The set $T(A)$ equipped with these
shift transformations is called the \emph{full shift} of
infinite trees over $A$.

A \emph{pattern} of a tree $t$ is a restriction of $t$ to a subset $L$ of $\Sigma^*$, 
where $L$ is called the \emph{support} of the pattern.
A \emph{block} of \emph{height} $k \geq 1$
on $A$ is a function $b \colon \Sigma^{k-1} \to A$.
We allow the block of height $0$ denoted by $\varepsilon$.
If $b$ is a block of height $k$, we let $b^{(k')}$ denote the subblock of $b$ of height $k'$ rooted at the root of $b$, if $k' \leq k$.

If $t$ is a tree on $A$ and $x \in \Sigma^*$ is a node, we let $t^{(k)}_x$ denote the block of height $k$ rooted at $x$, that is, the block $b$ such that $t(xy) = b(y)$ for any $y \in \Sigma^{k-1}$. 
We let $t^{(k)}$ denote the block of height $k$ at the root of $t$.
We say that a block $b$ \emph{is a block of height $k$ of a tree} $t$ if there is a node $x \in \Sigma^*$
such that $b = t^{(k)}_x$; otherwise, $b$ \emph{avoids} $t$. The empty block is a block of any tree.
If $X$ is a tree-shift, let $\cB_k(X)$ be the set of blocks of height $k$ of trees of $X$,
and $\cB(X)$ be the set of all blocks of trees of $X$.

If $b$ is a block of height $k \geq 2$ and $i \in  \Sigma$, we let $\sigma_i(b)$ denote its subblock 
of height $k{-}1$ rooted at $i$. 
A block $b$ of height $k\geq 2$ is denoted by $(b_\varepsilon, \sigma_1(b), \sigma_0(b))$.

A \emph{tree-shift} $X$ over $A$ is the set $\XS_F$, where $F$ is a set of blocks, consisting of all trees that avoid each block in $F$. 
A \emph{tree-shift of finite type} $X$ in $T(A)$ is a tree-shift $\XS_F$, where $F$ is a finite set of blocks.

Let $X$ be a tree-shift over the alphabet $A$, and let $B$ be another alphabet.
Given a nonnegative integer $n$,
a \emph{block map}
is a map $f\colon \cB_{n}(X)\to B$.
The \emph{sliding block code}
defined by $f$ is the map $\varphi:X\to T(B)$ defined by
$\varphi(t)=t'$ if for every $x \in \Sigma^*$,
\begin{displaymath}
  t'_x=f(t^{(n)}_x).
\end{displaymath}

The integer $n{-}1$ is the \emph{height anticipation} of $\varphi$.
The set $\varphi(X)$ is tree-shift.
A \emph{conjugacy} is an invertible sliding block code. Its inverse is also a sliding block code.

Let $X$ be a tree-shift over $A$ and let $k \geq 1$ be an integer. 
The map $\gamma_k \colon X\to T(\cB_k(X))$ defined for $t\in X$ by $t'=\gamma_k(t)$
if for every $x \in \Sigma^*$, 
\begin{equation}
t'_x = t^{(k)}_x,
\end{equation}
is the $k$-th \emph{higher block code} on $X$.
One also says that $\gamma_k$ is a \emph{coding by overlapping blocks} of length $k$.
The set $X^{(k)}=\gamma_k(X)$ is a tree-shift over $\cB_k(X)$, called the $k$-th
\emph{higher block shift} of $X$, or the $k$-th
\emph{higher block presentation}
of~$X$.

The following result is well known.
\begin{proposition}
The higher block code $\gamma_k \colon X\to T(\cB_k(X))$ is an isomorphism of tree-shifts
and the inverse of $\gamma_k$ is the block map $\pi_k \colon T(\cB_k(X)) \to T(A)$
defined by $g \colon \cB_k(X) \to A$ with $g(b ) = b(\varepsilon)$ (the letter label of the root of the block b).
\end{proposition}

\begin{example}\label{treeEx}

Let $T = T(A)$, where $A =\{a, b\}$ is the set of complete binary trees with nodes labeled on $\{a, b \}$.
A block of height $2$ in $\mathcal{B}_2(T)$ is denoted $b = (b_\varepsilon, b_1, b_0)$.
Let $F$ be the set of blocks of height $2$ containing $(b, b, b)$, $(b, b, a)$, and $(b, a, b)$.
Hence, no path in any tree in $\XS_F$ has two consecutive $b$'s.
In Figure \ref{figure.tree}, the child $x0$ of a node $x$ is represented under $x$ on the right,
and the child $x1$ of a node $x$ is represented under $x$ on the left. 

\tikzset{node distance=1cm, 
            every state/.style={ 
                  semithick,
                  minimum size=0.0mm},
            double distance=2pt, 
            every edge/.style={ 
                   draw,
                   -,
                   auto,
            semithick}}
\begin{figure}[htpb]
\centering
\begin{tikzpicture} 
      \node[state] (e) at (0,0) {$b$};
      \node[state] (0) at (1,-1) {$a$};
      \node[state] (1) at (-1,-1) {$a$};
      \node[state] (00) at (1.5,-2) {$a$};
      \node[state] (01) at (0.5,-2) {$b$};
      \node[state] (10) at (-0.5,-2) {$b$};
      \node[state] (11) at (-1.5,-2) {$b$};
      \node[] (000) at (1.75,-3) {};
      \node[] (001) at (1.25,-3) {};
      \node[] (010) at (0.75,-3) {};
      \node[] (011) at (0.25,-3) {};
      \node[] (100) at (-0.25,-3) {};
      \node[] (101) at (-0.75,-3) {};
      \node[] (110) at (-1.25,-3) {};
      \node[] (111) at (-1.75,-3) {};
      \draw (e) edge node {} (0);
      \draw (e) edge node {} (1);
      \draw (0) edge node {} (00);
      \draw (0) edge node {} (01);
      \draw (1) edge node {} (10);
      \draw (1) edge node {} (11);
      \draw[dashed] (00) edge node {} (000);
      \draw[dashed] (00) edge node {} (001);
      \draw[dashed] (01) edge node {} (010);
      \draw[dashed] (01) edge node {} (011);
      \draw[dashed] (10) edge node {} (100);
      \draw[dashed] (10) edge node {} (101);
      \draw[dashed] (11) edge node {} (110);
      \draw[dashed] (11) edge node {} (111);
\end{tikzpicture} 
\caption{A tree in $\XS_\mathcal{F}$.} \label{figure.tree}
\end{figure}
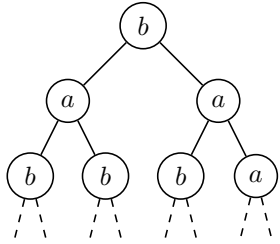
\end{example}


\subsection{Top-down tree automata and out-splitting} 

To simplify the notation, we consider binary trees below, i.e., we assume that $\Sigma= \{0, 1\}$. However, all results hold for any $d$, particularly for $d = 1$, extending the results for one-sided shifts of sequences.

A \emph{(top-down) tree automaton} on an alphabet $A$ is a structure $\cA = (Q, \Delta)$, where $Q$ is a finite set of states and $\Delta$ a finite set of transitions $(p, a) \to(q, r)$, with $p, q, r \in Q$, $a \in A$, $a$ being the label of the transition (see for instance \cite{Tata2021}).
A transition $(p, a) \to(q, r)$, also denoted $p \xrightarrow{a}(q, r)$,
is said to be an \emph{outgoing transition} of $p$,
a \emph{left incoming transition} of $q$ and a \emph{right incoming transition} of $r$. 

A tree $t$ labeled on $A$ is \emph{accepted} by $\cA$ if there is a tree $u$ on $Q$ such that, for any $x \in \Sigma^*$, 
$(u_x, t_x) \to (u_{x1}, u_{x0})$ belongs to $\Delta$. Such a tree $u$ is called a \emph{computation tree} for $t$.
A tree is a computation tree of $\cA$ if it is a computation tree of some tree accepted by $\cA$.
The set of trees accepted by $\cA$ is a tree-shift denoted by $X_\cA$.

A tree automaton is \emph{trim} if each state of the automaton has at least one outgoing transition. 

An \emph{edge tree automaton} is a tree automaton such that all transitions have distinct labels. 
An \emph{edge tree-shift} is the set of trees accepted by an edge tree automaton. 

The following result is well known.
\begin{proposition}
A tree-shift of finite type is conjugate to an edge tree shift.
\end{proposition}

\begin{proof}
Let $X = \XS_F$ be a tree shift of finite type, where $F$ is a set of blocks of height $k{+}1$ with $k \geq 1$. It is conjugate to the tree shift $Y =  \gamma_k(X)$ on $\cB_k(X)$. Let us show that $Y$ is an edge tree shift. 
Let  $\cA = (\cB_k(X), \Delta)$ be the tree automaton defined as follows.
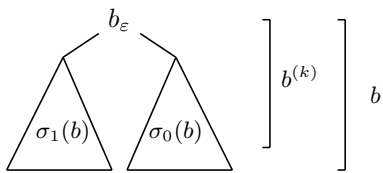
\begin{figure}[htpb]
\centering
   \tikzset{node distance=0cm, 
            every edge/.style={ 
                   draw,
                   -,
                   auto,
            semithick}}
\begin{tikzpicture} 
      \node[] (e) at (0,0) {$b_\varepsilon$};
      \coordinate (0) at (0.75,-0.5) {};
      \coordinate (1) at (-0.75,-0.5);
      \coordinate (00) at (1.5,-2);
      \coordinate (01) at (0.1,-2);
      \coordinate (10) at (-0.1,-2);
      \coordinate (11) at (-1.5,-2);
      \draw (e) edge node {} (0);
      \draw (e) edge node {} (1);
      \draw (0) edge node {} (00);
      \draw (0) edge node {} (01);
      \draw (1) edge node {} (10);
      \draw (1) edge node {} (11);
      \draw (00) edge node {} (01);
      \draw (10) edge node {} (11);
      \node at (0.75,-1.5) {{\small $\sigma_0(b)$}};
      \node at (-0.75,-1.5) {{\small $\sigma_1(b)$}};
     \coordinate (A) at (2,0) {};
     \coordinate (B) at (2,-1.7) {};
     \coordinate (AA) at (1.9,0) {};
     \coordinate (BB) at (1.9,-1.7) {};
     \draw (A) edge node {} (B);
      \draw (A) edge node {} (AA);
       \draw (B) edge node {} (BB);
     \node at (2.4,-0.75) {{\small $b^{(k)}$}};
     \coordinate (C) at (3,0) {};
     \coordinate (D) at (3,-2) {};
     \coordinate (CC) at (2.9,0) {};
     \coordinate (DD) at (2.9,-2) {};
     \draw (C) edge node {} (D);
      \draw (C) edge node {} (CC);
       \draw (D) edge node {} (DD);
     \node at (3.4,-1) {{\small $b$}};
\end{tikzpicture} \label{figure.treeautomaton}
\caption{The block $b \notin F$ of height $k{+}1$.}
\end{figure}

For each block $b = (b_\varepsilon, \sigma_1(b), \sigma_0(b)) \in \cB_{k+1}(X) \setminus  F$,
there is a transition $(b^{(k)}, b)  \to (\sigma_1(b), \sigma_0(b))$ in $\Delta$.
Note that all transitions have distinct labels~$b$.
The automaton $\cA$ is an edge tree automaton that accepts~$\XS_F$.
\end{proof}
Since an edge tree automaton has all its transitions having distinct labels, we may represent all transitions $(p, a) \to(q, r)$ going from a state $p$ to a pair of states $(q, r)$ by $p \xrightarrow{n}(q, r)$, where $n \geq 1$ is the number of transitions from $p$ to $(q, r)$. 
We define the \emph{transition matrix} $M$ of such an automaton by $M_{p, (q,r)} = n$ if there is a transition $p \xrightarrow{n}(q, r)$.

We define the notion of \emph{out-splitting} of an edge tree automaton as follows.

An out-splitting transforms an edge tree automaton $\cA=(Q, \Delta)$ into an edge tree automaton $\cA'=(Q', \Delta')$ that is obtained from
a partition of the set of transitions going out of a state $s$  into two nonempty sets $\Delta_1$ and $\Delta_2$. 
The set of states of $\cA'$ is $Q' = Q \setminus \{s\} \cup \{s_1, s_2\}$. 
Let $\mu(q)$ denote the set $\{q\}$ if $q \neq s$, and $\mu(s)$ denotes  $\{s_1,s_2\}$.
The transitions of $\Delta'$ are the following:
\begin{itemize}
\item $(p, a_{(q', r')}) \to (q', r')$ if $(q',r') \in \mu(q) \times \mu(r), p \neq s, (p, a) \to (q, r) \in \Delta$;
\item $(s_i,a_{(q', r')})  \to (q', r')$  if $(q',r') \in \mu(q) \times \mu(r)$, $(s, a) \to (q, r) \in \Delta_i$.
\end{itemize}
Note that all transitions of $\cA'$ have distinct labels since all transitions of $\cA$ have distinct labels.


The inverse operation is called an (elementary) out-merging. We define the notion of \emph{general out-splitting} similarly. The outgoing transitions from a given state can be partitioned into arbitrary many subsets instead of just two, and the partitioning can occur simultaneously at all of the states instead of just one.  

Let $\cA'=(Q',\Delta')$ be an edge tree automaton and $M'$ its transition matrix. 
Let $(Q'_1, Q'_2, \ldots, Q'_k)$ be a partition of $Q'$ into classes
such that if $q_i, q_j$ belong to the same class, then for each state $p$ of $Q'$
\begin{enumerate}
\item the columns of indices $(q_i,p)$ and $(q_j,p)$ of $M'$ are identical, \label{condition.merge.1}
\item the columns of indices $(p, q_i)$ and $(p, q_j)$ of $M'$ are identical.\label{condition.merge.2}
\end{enumerate}
We can define an edge tree automaton $\cA=(Q, \Delta)$ of transition matrix $M$ obtained by merging all states of each $Q'_i = \{q_{i,1}, \ldots q_{i,k_i}\}$ into a single state
$q_{i,1}$ whose row in $M$ is obtained by summing of the rows $q_{i,1}, \ldots q_{i,k_i}$ of $M'$ and removing the columns $(p, q_{i,2}), \cdots, (p, q_{i,k_i})$ and $(q_{i,2}, p), \cdots, (q_{i,k_i}, p)$ for each $p \in Q'$. The labels of the transitions of $\cA'$ are all distinct.

The automaton $\cA$ is called a \emph{general amalgamation} of $\cA'$. When only two states of $\cA'$ are merged, $\cA$ is called an \emph{elementary amalgamation. 
An edge tree automaton that can be obtained 
throw a sequence of elementary out-mergings is called an \emph{amalgamation} of $\cA'$.}

\begin{proposition} \label{proposition.tree.splitting}
Suppose $X$ is the edge tree-shift defined by an edge tree automaton $\cA$, and $Y$ is the edge tree-shift defined by the edge tree automaton $\cA'$ that is an out-splitting automaton of $\cA$.  Then $X$ and $Y$ are conjugate.
\end{proposition}
\begin{proof}
Let us assume that $\cA'= (Q', \Delta')$ is obtained from $\cA= (Q, \Delta)$ by splitting a state $s$ into two states $s_1, s_2$, and partitioning the transitions going out of $s$ into $\Delta_1$ and $\Delta_2$. 
Recall that all transitions of $\cA$ and $\cA'$ have distinct labels. 

We define a sliding block code $\varphi \colon X \to Y$ from a block map 
$f \colon \cB_2(X) \to \cB(Y)$
as follows.

For each block $(a, b, c)$ of height $2$ of $X$, there are unique transitions 
$(p, a) \to (q, r)$,  $(q, b) \to (t, u)$, and $(r, c) \to (v, w)$ in $\cA$.
Below, let $i, j \in \{ 1,2 \}$.

If $p,q,r \neq s$, then we set $f(a, b, c) = a_{(q, r)}$.

If $p, q \neq s$, $r = s$ and $(s, c) \to (v, w) \in \Delta_i$, then we set $f(a, b, c) = a_{(q, s_i)}$.

If $p, r \neq s$, $q = s$ and $(s, b) \to (t, u) \in \Delta_i$, then we set $f(a, b, c) = a_{(s_i, r)}$.

If $p \neq s$, $q = r = s$, $(s, b) \to (t, u) \in \Delta_i$, and $(s, c) \to (v, w) \in \Delta_j$,
then we set $f(a, b, c) = a_{(s_i, s_j)}$.

If $p = s$, and $q, r \neq s$ and $(p,a) \to (q, r) \in \Delta_i$, then we set $f(a, b, c) = a_{(q, r)}$.

If $p = q = s$, $r \neq s$, and $(s, b) \to (t, u) \in \Delta_i$,
then we set $f(a, b, c) = a_{(s_i, r)}$.

If $p = r = s$, $q \neq s$, and $(s, c) \to (v, w) \in \Delta_i$,
then we set $f(a, b, c) = a_{(q, s_i)}$.

If $p = q = r = s$, $(s, b) \to (t, u) \in \Delta_i$, and $(s, c) \to (v, w) \in \Delta_j$,
then we set $f(a, b, c) = a_{(s_i, s_j)}$.

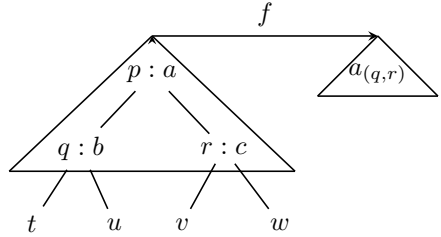
\begin{figure}[htpb]
\centering
   \tikzset{node distance=0cm, 
            every edge/.style={ 
                   draw,
                   auto,
            semithick}}
\begin{tikzpicture} 
      \node[] (e) at (0,0) {$p: a$};
      \node[] (0) at (0.95,-1) {$r: c$};
      \node[] (1) at (-0.95,-1) {$q: b$};
      \node[] (00) at (1.7,-2) {$w $};
      \node[] (01) at (0.4,-2) {$v $};
      \node[] (10) at (-0.5,-2) {$u $};
      \node[] (11) at (-1.6,-2) {$t $};
      \draw (e) edge node {} (0);
      \draw (e) edge node {} (1);
      \draw (0) edge node {} (00);
      \draw (0) edge node {} (01);
      \draw (1) edge node {} (10);
      \draw (1) edge node {} (11);
     \coordinate (A) at (0,0.5) {};
     \coordinate (B) at (1.9,-1.3) {};
     \coordinate (C) at (-1.9,-1.3) {};
     \draw (A) edge node {} (B);
      \draw (B) edge node {} (C);
       \draw (A) edge node {} (C);
       \node[] (f) at (3,0) {$a_{(q,r)}$};
   \coordinate (AA) at (3,0.5) {};
     \coordinate (BB) at (3.8,-0.3) {};
     \coordinate (CC) at (2.2,-0.3) {};
     \draw (AA) edge node {} (BB);
      \draw (BB) edge node {} (CC);
       \draw (AA) edge node {} (CC);
       \draw[->,>=stealth, auto] (A) edge node {$f$} (AA);
\end{tikzpicture} \label{figure.treesblockmap}
\caption{The $2$-block map $f$ when $p, q, r \neq s$.}
\end{figure}
The sliding block map $\varphi$ is clearly a conjugacy whose inverse is the sliding block map $\psi$ defined
by the $1$-block map $g \colon \cB(Y) \to \cB(X)$ as follows: $g(a_{(q',r')}) = a$.

Let us show that $\varphi(X) = Y$. 
Let $z \in X$. We first check that $\varphi(z)$ is accepted by $\cA'$. 
By construction, for each $x \in \Sigma^*$, if $z_x^{(2)} = (a, b, c)$ and $f(a, b, c) = a_{(q', r')}$ with $q', r' \in Q'$, then there is $p \in Q'$ such that $(p, a) \to (q', r') \in \Delta'$. Hence, $\varphi(z)$ is accepted by $\cA'$. 
Conversely, if $z' \in Y$, let us show that $z' \in \varphi(X)$.
Let $u'$ be a computation tree of $z'$ in $\cA'$. Then, for each $x \in \Sigma^*$, 
$(u'_x, z_x) \to (u'_{x1}, u'_{x0})$ belongs to $\Delta'$. 
Thus, $z_x = a_{(q', r')}$ for some letter $a$ of the alphabet of $\cA$ and some states $q', r' \in Q'$.
By construction, this implies that $(u'_{x1}, u'_{x0}) = (q', r')$, $g(z_x) = a$, and
$(\pi(u'_x), a) \to (\pi(q'), \pi(r'))$ belongs to $\Delta$,
where $\pi(q'') = q''$ if $q'' \neq s_1, s_2$ and $\pi(s_1) = \pi(s_2) = s$.
Hence, $z = \psi(z')$ is accepted by $\cA$ and $\varphi(z) = z'$.

\end{proof}

For tree-shifts, William's theory for one-sided shifts of finite type holds. 
The same result was obtained in \cite{AubrunBeal2012} with top-down tree automata.
\begin{proposition}  \label{proposition.tree.williams1}
If $\cA$ and $\cB$ are amalgamations of a common edge tree-auto\-maton, then they have a common amalgamation.
\end{proposition}
\begin{proof}
Assume that $\cA$ (respectively $\cB$) is obtained by merging states $q_i, q_j$  (respectively $q'_i, q'_j$) of an edge tree-automaton $\cC$
into a state denoted $q_{ij}$  (respectively $q'_{ij}$).  Thus, the columns of indices $(q_i,p)$ and $(q_j ,p)$ (respectively $(p, q_i)$ and $(p, q_j)$)
of the adjacency matrix of $\cC$ are identical, and the columns of indices $(q'_i,p)$ and $(q'_j ,p)$  (respectively $(p, q'_i)$ and $(p, q'_j)$)
of the adjacency matrix of $\cC$ are identical. 

Assume that $q'_i \neq q_i, q_j$, $q'_j \neq q_i, q_j$. If $p \neq q_i, q_j$,
then the columns of indices $(q'_i,p)$ and $(q'_j ,p)$  (respectively $(p, q'_i)$ and $(p, q'_j)$)
of the adjacency matrix of $\cC$ remain identical in the adjacency matrix of $\cA$. If $p = q_i$,
the columns of indices $(q'_i,q_i)$ and $(q'_j ,q_i)$ 
(respectively $(q_i, q'_i)$ and $(q_i, q'_j)$) of the adjacency matrix of $\cC$ become the columns of indices 
$(q'_i,q_{ij})$ and $(q'_j ,q_{ij})$ (respectively $(q_{ij}, q'_i)$ and $(q_{ij}, q'_j)$) 
of the adjacency matrix of $\cA$
and remain identical. Thus, one can merge the states  $q'_i, q'_j$ in $\cA$ leading to an automaton $\cD$.

Let $M$ be the adjacency matrix of $\cC$ and $N$ the adjacency matrix of $\cD$.
If $q \neq q_{ij}, q'_{ij}$, 
we have $N_{q, (r, t)} = M_{q, (r', t')}$ for each $r', t'$ such that
$\pi(r') = r, \pi(t') = t$, where $\pi(s) = s$ for each $s \neq q_i, q_j, q'_i, q'_j$, $\pi(q_i) = \pi(q_j) = q_{ij}$, and $\pi(q'_i) = \pi(q'_j) = q'_{ij}$.
We also have $N_{q_{ij}, (r, t)} = M_{q_i, (r', t')} + M_{q_j, (r', t')}$, 
$N_{q'_{ij}, (r, t)} = M_{q'_i, (r', t')} + M_{q'_j, (r', t')}$ for each $r', t'$ such that
$\pi(r') = r, \pi(t') = t$.

Similarly, one can merge the states $q_i, q_j$ in $\cB$ and get an automaton $\cD'$ with the same 
transition matrix as $\cD$.

Assume now that $q'_i =  q_i$ (and $q'_j \neq q_i, q_j$),  the columns of indices $(q_i,p)$,  $(q_j ,p)$  and $(q'_j,p)$ (respectively $(p, q_i)$,  $(p, q_j)$  and $(p, q'_j)$)
of $M$
are equal. Thus, in $\cA$, the columns of indices $(q_{ij},p)$ and $(q'_j ,p)$ (respectively $(p, q_{ij})$ and $(p, q'_j)$) are equal.
Therefore, it is possible to merge the states $q_{ij}$ and $q'_j$ of $\cA$ into a state denoted $(q_{ij}, q'_j)$ and get an automaton $\cD$. 

If $q \neq (q_{ij}, q'_j)$, 
we have $N_{q, (r, t)} = M_{q, (r', t')}$ for all $r', t'$ such that
$\pi(r') = r, \pi(t') = t$, where $\pi(s) = s$ for each $s \neq q_i, q_j, q'_j$, 
$\pi(q_i) = \pi(q_j) = 
\pi(q'_j) = (q_{ij},q'_j)$. 
We also have $N_{(q_{ij}, q'_j), (r, t)} = M_{q_i, (r', t')} + M_{q_j, (r', t')} + M_{q'_j, (r', t')}$, 
for all $q', r'$ such that
$\pi(r') = r, \pi(t') = t$.

Similarly, it is possible to merge the states $q'_{ij}$ and $q_j$ of $\cB$ into a state denoted $(q'_{ij}, q_j)$
and obtain an automaton~$\cD'$ identical to $\cD$ after renaming $(q'_{ij}, q_j)$ by $(q_{ij}, q'_j)$.
All other cases run identically to the one above, or are trivial.
\end{proof}

As a consequence, given an edge tree automaton $\cC$, there is a unique edge tree automaton, up to a renaming of the states, 
obtained by performing out-mergings until we cannot perform any more.
This edge tree automaton is called the \emph{total amalgamation} of $\cC$. 

\begin{example}
Let $\cA'=(Q', \Delta')$ be the edge tree automaton with $Q' = \{q_1, q_2, q_3\}$ and transitions
\begin{align*}
&q_1  \xrightarrow{1}(q_1, q_1), &q_1  \xrightarrow{1}(q_1, q_2), &&q_1  \xrightarrow{1}(q_2, q_1), &&q_1  \xrightarrow{1}(q_2, q_2),\\
&q_2  \xrightarrow{1}(q_3, q_3), &&&\\
&q_3  \xrightarrow{1}(q_1, q_1), &q_3  \xrightarrow{1}(q_1, q_2), &&q_3  \xrightarrow{1}(q_2, q_1), &&q_3  \xrightarrow{1}(q_2, q_2).
\end{align*}
Recall that the integer $n$ above a transition $p  \xrightarrow{n}(q, r)$ means that there are $n$ transitions from $p$ to $(q, r)$, and that all transitions have distinct labels.

The transition matrix of $\cA'$ is 
\begin{align*}
M' &= 
\begin{blockarray}{cccccc}
 & (q_1, q_1) & (q_1, q_2) & (q_2, q_1) & (q_2, q_2) & (q_3, q_3) \\
\begin{block}{c(ccccc)}
  q_1 & 1 & 1 & 1 & 1 & 0 \\
  q_2 & 0 & 0 & 0 & 0 & 1 \\
  q_3 & 1 & 1 & 1 & 1 & 0 \\
\end{block}
\end{blockarray}.
\end{align*}
Its total amalgamation is obtained by merging states $q_1$ and $q_2$ into $q_{12}$ with transition matrix is
\begin{align*}
M &= 
\begin{blockarray}{ccc}
 & (q_{12}, q_{12}) & (q_3, q_3) \\
\begin{block}{c(cc)}
  q_{12} & 1 & 1  \\
  q_3    & 1 & 0  \\
\end{block}
\end{blockarray}.
\end{align*}
\end{example}
 
\begin{example}
Let $\cA'=(Q', \Delta')$ be the tree automaton with $Q' = \{q_1, q_2\}$ and transitions
\begin{align*}
&q_1  \xrightarrow{1}(q_1, q_1), &q_1  \xrightarrow{1}(q_1, q_2), &&q_1  \xrightarrow{1}(q_2, q_1), &&q_1  \xrightarrow{1}(q_2, q_2),\\
&q_2  \xrightarrow{1}(q_1, q_1), &q_2  \xrightarrow{1}(q_1, q_2), &&q_2  \xrightarrow{1}(q_2, q_1), &&q_2  \xrightarrow{1}(q_2, q_2).
\end{align*}
Its transition matrix is 
\begin{align*}
M' &= 
\begin{blockarray}{ccccc}
 & (q_1, q_1) & (q_1, q_2) & (q_2, q_1) & (q_2, q_2) \\
\begin{block}{c(cccc)}
  q_1 & 1 & 1 & 1 & 1 \\
  q_2 & 1 & 1 & 1 & 1 \\
\end{block}
\end{blockarray}.
\end{align*}
Its total amalgamation is obtained by merging states $q_1$ and $q_2$ into $q_{12}$ with transition matrix is
\begin{align*}
M &= 
\begin{blockarray}{cc}
 & (q_{12}, q_{12})\\
\begin{block}{c(c)}
  q_{12} & 2   \\
\end{block}
\end{blockarray}.
\end{align*}
\end{example}

The following result generalizes Proposition \ref{propositionWilliams2} to tree-shifts. It was established in \cite{AubrunBeal2012} using top-down tree automata and under hypotheses slightly stronger than trimness.

\begin{proposition}  \label{proposition.tree.williams2}
Let $\cA$ and $\cB$ be trim edge tree automata that define edge tree-shifts $X$ and $Y$ respectively. Then $X$
and $Y$ are conjugate if and only if $\cA$ and $\cB$ have the same total
amalgamation.
\end{proposition}

\begin{proof}
If $\cA$ and $\cB$ have the same total amalgamation
then they are both out-splittings of that amalgamation, 
hence by Proposition \ref{proposition.tree.splitting}, they are conjugate.

Conversely, let 
$\cA=(Q, \Delta)$ and $\cB=(Q', \Delta')$ be two trim edge tree automata that define edge tree-shifts $X$ and $Y$ respectively.
Let us assume that there is conjugacy $\varphi \colon X \to Y$ with height anticipation $n{-}1$ defined by the block map $f\colon \cB_{n}(X)\to B$, where $B$ is the alphabet of $Y$.
Assume that $\varphi^{-1}$ is the sliding block code with height anticipation $n'{-}1$ defined by the block map $f'\colon \cB_{n'}(Y)\to A$, where $A$ is the alphabet 
of $X$. Without loss of generality, we may suppose that $n, n' \geq 1$.

A block $b$ of height $n$ can be written $b = (a, \sigma_1(b), \sigma_0(b))$, where $\sigma_i(b)$ are blocks of height $n{-}1$, and $a = b(\varepsilon)$. 
 
We define a tree automaton $\cC = (\cB_n(X) \times \cB_{n'}(Y), \Gamma)$ labeled in $A \times B$ as follows. 
There is a transition 
\[
(b, b') \xrightarrow{(a, f(b))} ((c, c'), (d, d'))
\]
in $\cC$ if and only if
\begin{itemize}
\item $b = (a, \sigma_1(b), \sigma_0(b))$ with $c^{(n-1)} =  \sigma_1(b)$ and $d^{(n-1)} =  \sigma_0(b)$.
\item $b' = (f(b), \sigma_1(b'), \sigma_0(b'))$ with $c'^{(n-1)} =  \sigma_1(b')$ and $d'^{(n-1)} =  \sigma_0(b')$,
\item $(b, c, d) \in \cB_{n+1}(X)$, $(b', c', d') \in \cB_{n+1}(Y)$.
\end{itemize}

By construction, $\cC$ is trim.
Let $\textbf{t}$ be a tree labeled on $A \times B$ accepted by~$\cC$. 
Let $t$ (respectively $t'$) be the tree obtained from $t$ by taking the first (respectively second) component of the labels of the nodes of $\textbf{t}$.
The tree $t$ is called the \emph{input (respectively output) tree} of $\textbf{t}$.
By construction, $t' = \varphi(t)$. 
Conversely, if  
$t$ is a tree of $X$ and $t' = \varphi(t)$, 
there is a unique computation tree $\textbf{u}$ of $\cC$ accepting $\textbf{t}$. 
This computation tree starts at the state $(b , b')$, where $b = t^{(n)}$ and $b' = t'^{(n')}$.

Let us show that the automaton $\cA$ can be obtained from $\cC$ through a sequence of 
out-mergings. 

We can perform a general amalgamation of $\cC$ by merging all states $(b, b')$, $(b , c')$
sharing the same block $b$ and such that $b'^{(n'-1)} = c'^{(n'-1)}$. 

Indeed, let $(b , b')$, $(b , c')$ be two states of $\cC$ such that $b'^{(n'-1)} = c'^{(n'-1)}$.
Let 
\[(e, e') \xrightarrow{(e(\varepsilon), f(e))} ((b , b'), (g , g'))
\]
be a left incoming transition of $(b, b')$
in $\cC$. By construction, there is also a transition in $\cC$
\[ (e, e') \xrightarrow{e(\varepsilon), f(e))} ((b, c'), (g, g')).
\]
The same property holds for right incoming transitions of $(b, b')$.
We assume that all states $(b, b')$ sharing the same blocks $b$ and $b'^{(n'-1)}$
have been merged into a state denoted by $(b, b'^{(n'-1)})$ in an automaton denoted  $\cC(n, n'-1)$.
After this amalgamation, there is a transition in $\cC(n, n'-1)$
\[ (e, e'^{(n'-1)})\xrightarrow{(e(\varepsilon), f(e))} ((b, b'^{(n'-1)}), (g, g'^{(n'-1)})), 
\]
if and only if
\begin{itemize}
\item $e = (e(\varepsilon), \sigma_1(e), \sigma_0(e))$ with $b^{(n-1)} =  \sigma_1(e)$ and $g^{(n-1)} =  \sigma_0(e)$.
\item $e'^{(n'-1)} = (f(e), \sigma_1(e'^{(n'-1)}), \sigma_0(e'^{(n'-1)}))$ with $b'^{(n'-2)} =  \sigma_1(e'^{(n'-1)})$ and $g'^{(n'-2)} =  \sigma_0(e'^{(n'-1)})$.

\end{itemize}
Hence,  if $b'^{(n'-2)} = c'^{(n'-2)}$, then 
$(b, b'^{(n'-1)})$ has a left (respectively right) incoming transition from $(e, e'^{(n'-1)})$ if and only if $(b,  c'^{(n'-1)})$ has a left (respectively right) incoming transition from $(e, e'^{(n'-1)})$.

For each $1 \leq k \leq n'$, we perform a general amalgamation of $\cC(n, k)$ into $\cC(n, k-1)$
by merging all states 
$(b, b'^{(k)} )$, $(b , c'^{(k)})$
sharing the same block $b$ and such that $b'^{(k-1)} = c'^{(k-1)}$
(where $b^{(0)} =\varepsilon$ is the empty block).
 Indeed, all states $(b, b'^{(k)})$, $(b, c'^{(k)})$ with 
$b'^{(k-1)} = c'^{(k-1)}$ have the same left (respectively right) incoming transitions.
The states of $\cC(n, 0)$ are pairs $(b, \varepsilon)$, where $\varepsilon$ is the empty block.
We define the automaton $\cD(n)$ obtained from $\cC(n, 0)$ by keeping the first components of the states and of the edges' labels.


For $2 \leq k \leq n$, we perform a general amalgamation of $\cD(k)$ into $\cD(k-1)$ by merging all states $b^{(k)} $, $ c^{(k)} $ such that $b^{(k-1)} = c^{(k-1)}$. 
Indeed, states $ b^{(k)}$, $ c^{(k)} $ such that $b^{(k-1)} = c^{(k-1)}$
have the same left (respectively right) incoming transitions in $\cD(k)$.

Finally, the automaton $\cD(1)$ is either $\cA$ or an out-splitting of $\cA$. Indeed, we may perform a general out-splitting $\cD'$ of $\cA$ by partitioning the transitions going out of a state $p$ into singletons. 
If $p$ is a state of $\cA$, it is split into states $(p, a)$, for each label $a$ of a transition going out of $p$. Since $\cA$ is trim, if $(p,a) \to (q, r)$ is a transition of $\cA$,
then $((p, a), a) \to ((q, b), (r, c))$ is a transition of $\cD'$ for all letters $b, c$ such that 
$(a, b, c) \in \cB_2(X)$. And $\cD' = \cD(1)$ up the renaming of the states.

We now show that the automaton $\cB$ can be obtained from $\cC$ through a sequence of amalgamations.
We just exchange the roles played by $\cA$ and $\cB$, that is,  the roles played by the first and the second components of the states
and of the labels of the transitions. Indeed, if 
\[ (e, e' ) \xrightarrow{(e(\varepsilon), f(e))} ((b, c'), (g, g')).
\]
is a transition of $\cC$, then $e(\varepsilon) = g(e')$, where $g$ is the block map defining $\varphi^{-1}$, and $f(e) = e'(\varepsilon)$.

Hence, $\cA$ and $\cB$ are amalgamations of $\cC$.
By Proposition \ref{proposition.tree.williams1}, $\cA$ and $\cB$ have the same total amalgamation.
\end{proof}
Note that the result holds even if the edge tree-shifts are not irreducible (see \cite{AubrunBeal2013} for a definition).
We only assume they are defined by trim edge automata.

\section{Hom tree-shifts} \label{section.hom.tree.shifts}
Let $G=(V, E)$ be a simple (unlabeled) undirected graph (where self-loops are allowed). The \emph{Hom tree shift} defined by $G$, denoted $X(G)$, is the set of trees $t$ labeled on $V$ such that each path of $t$ is a path of $G$. In other words, each sequence $(y_n)_{n \geq 0} = (t_{x_n})_{n \geq 0}$ such that $x_{n+1} \in x_{n}\Sigma^*$, is a sequence of vertices of $G$ such that $(y_n, y_{n+1})$ is an edge of $G$.

A Hom tree-shift defined by $G =(V, E)$ is a tree-shift of finite type. Indeed, it is equal to $\XS_F$, where $F$ is the set of blocks $b$ of height $2$ on $V$ such that $(b(\varepsilon), b(i))$ is not an edge of $G$ for at least one $i \in \Sigma$.

It is accepted by the tree automaton $\cB = (V, \Delta)$
whose transitions are $(p, p) \to (q, r)$, where $(p, q), (p, r) \in E$. Note that the label $p$ of this transition is equal to its starting state.

Note that if $\cB$ has the transition $(p, p) \to (q, r)$, or if it has the transitions
$(p, p) \to (q, q)$ and $(p, p) \to (r, r)$,
then it must have all transitions:
\begin{itemize}
    \item $(p, p) \to (q, r)$
    \item $(p, p) \to (q, q)$
    \item $(p, p) \to (r, r)$
    \item $(p, p) \to (r, q)$
\end{itemize}
Moreover, since $G$ is an undirected graph, if $\cB$ has the transition $(p, p) \to (q, r)$, then it has transitions $(q, q) \to (p, p)$ and $(r, r) \to (p, p)$.

Let $\cA(G)$ be the edge tree automaton obtained from $\cB$ by labeling all transitions with distinct labels. 
Note that, if $(p, q, r) \in \cB_2(X(G))$, then $\cA(G)$ has a unique transition from $p$ to $(q, r)$.
\begin{lemma}
Let $G$ be a simple undirected graph.
The tree-shift $X(G)$ is conjugate to the edge tree-shift accepted by $\cA(G)$.    
\end{lemma}
\begin{proof}
Let $G = (V, E)$ be a simple undirected graph and $\cA(G)= (V, \Delta)$ be the above defined edge tree automaton. 
Let $A$ be the alphabet of $\cA(G)$.
Let $\varphi \colon X_{\cA(G)} \to X(G)$ be the sliding block code defined by the $1$-block map $f \colon A \to V$ with $f(a) = p$ 
if $p$ is the starting state of the unique transition labeled by $a$.
The map $\varphi$
 is a conjugacy from $X_{\cA(G)}$ to $X(G)$. Its inverse is the sliding block code defined by the $2$-block map $g \colon \cB_2(X(G)) \to A$ with 
 $g((p, q, r)) = a$, where $a$ is the label of the unique transition from $p$ to $(q, r)$.
\end{proof}

Note that we will usually omit the node labels from our notation whenever the labels are not directly used and when transitions have distinct labels.

We say that an edge tree automaton $\cA=(Q, \Delta)$ with adjacency matrix $M$ is \emph{regular} if each state $p$:
\begin{itemize}
\item if $M_{p, (q, r)} > 0$ and  $M_{p, (s, t)} > 0$, then $M_{p, (u, v)} > 0$ for all $u, v \in \{q, r, s, t\}$;
\item the non-null coefficients of the row of index $p$ of $M$ are equal;
\item if $M_{p, (q,r)} > 0$, then $M_{q, (p, p)} > 0$ and $M_{r, (p, p)}>0$.
\end{itemize}
These first two conditions are equivalent to the following: for each state $p$, there is a set of states $S(p)$ 
and a positive integer $n_p$ such that, for each states $q, r$, either $M_{p, (q, r)} = 0$ or  $M_{p, (q, r)} = n_p$,
and $M_{p, (q, r)} = n_p$ if and only if $q, r \in S(p)$.
Note that the edge tree automaton $\cA(G)$ is regular, but the edge tree automaton obtained by an out-splitting of $\cA(G)$ may not remain regular.

\begin{proposition} \label{proposition.tree.regular}
The amalgamation of a regular edge tree automaton is a regular edge tree automaton.
\end{proposition}
\begin{proof}
Assume that two states $q_i, q_j$ of a regular edge tree automaton $\cA$ are merged into a state $q_{ij}$ in an edge tree automaton $\cB$. Let $M, N$ be the adjacency matrices of $\cA$ and $\cB$, respectively.

Thus, the columns of indices $(q_i, p)$ and $(q_j, p)$ (respectively $(p, q_i)$ and $(p, q_j)$) of the adjacency matrix $M$ of $\cA$ are identical for each state $p$. 
For a state $p$ of $\cA$, let $\pi(p)$ denote the state $p$ if $p \neq q_i, q_j$, where $\pi(q_i) = \pi(q_j) = q_{ij}$.

Let us show that $\cB$ satisfies the last condition of regularity. 
If $N_{\pi(p), (\pi(q), \pi(r))}$ $> 0$, 
then there exist $p'\in V$ with $\pi(p') = \pi(p)$ such that for all $q', r' \in V$ 
such that $\pi(q') = \pi(q)$, $\pi(r') = \pi(r)$, we have $M_{p', (q', r')} = n_{p'} > 0$.
This implies that $M_{q', (p', p')} > 0$ and $M_{r', (p', p')} > 0$, implying
$N_{\pi(q), (\pi(p), \pi(p))} > 0$ and $N_{\pi(r), (\pi(p), \pi(p))} > 0$.

We now show that the first two conditions of regularity are verified.
For each state $p$, there is a set of states $S(p)$ and a positive integer $n_p$ such that, for each states $q, r$, either $M_{p, (q, r)} = 0$ or  $M_{p, (q, r)} = n_p$,
and $M_{p, (q, r)} = n_p$ if and only if $q, r \in S(p)$.
For each state $\pi(p)$ of $\cB$, we set $S(\pi(p)) = \cup_{q \in V \mid \pi(q) = \pi(p)} \pi(q) =  \pi(S(p))$, and $n_{\pi(p)} = \sum_{q \in V \mid \pi(q) = \pi(p)} n_q$.

Let us show that either $N_{\pi(p), (\pi(q), \pi(r))} = 0$ or $N_{\pi(p), (\pi(q), \pi(r))} = n_{\pi(p)}$,
and $N_{\pi(p), (\pi(q), \pi(r))}$ $= n_{\pi(p)}$ if and only if $\pi(q), \pi(r) \in S(\pi(p))$.

By construction, if $N_{\pi(p), (\pi(q), \pi(r))} > 0$, then there exists $p' \in V$ with $\pi(p') = \pi(p)$,
such that for all $q', r' \in V$ with $\pi(q') = \pi(q)$, $\pi(r') = \pi(r)$, we have $M_{p', (q', r')} = n_{p'} > 0$. 

If $\pi(p) \neq q_{ij}$, then $p' = p$ and $N_{\pi(p), (\pi(q), \pi(r))} = n_p = n_{\pi(p)}$.

If $\pi(p) =  q_{ij}$, then $p = q_i$ or $p = q_j$, and $p' = q_j$ or $p' = q_i$.

Thus, $M_{q', (p', p)} = M_{q', (p', p')} > 0$,
 $M_{r', (p', p)} > 0$, 
$M_{q', (p, p)} > 0$ and $M_{r', (p, p)} > 0$, implying
$M_{p, (q', q')} > 0$ and $M_{p, (r', r')} > 0$.
We get $q', r' \in S(p)$. 
Hence, $N_{\pi(p), (\pi(q), \pi(r))} 
= N_{\pi(p), (\pi(q'), \pi(r'))} = n_{q_i} + n_{q_j}  = n_{\pi(p)}$.

Finally, if $N_{\pi(p), (\pi(q), \pi(r))} > 0$,
there exists $p'\in V$ with $\pi(p') = \pi(p)$, such that for all $q', r' \in V$ with $\pi(q') = \pi(q)$, $\pi(r') = \pi(r)$, 
we have $M_{p', (q', r')} = n_{p'} > 0$. 
If $N_{\pi(p), (\pi(s), \pi(t))} > 0$, 
there exist $p''\in V$ with $\pi(p'') = \pi(p)$, such that for all $s', t' \in V$ with $\pi(s') = \pi(s)$, $\pi(t') = \pi(t)$, 
we have $M_{p'', (s', t')} = n_{p''} > 0$. 
This implies that $M_{s', (p'', p'')} > 0$ and  $M_{t', (p'', p'')} > 0$,
$M_{s', (p'', p)} > 0$ and  $M_{t', (p'', p)} > 0$,
and thus $M_{p, (s', s')} > 0$ and  $M_{p, (t', t')} > 0$.
Thus, $s', t' \in S(p)$. As a consequence, 
$N_{\pi(p), (u', v')} > 0$ for each $u',v' \in \{\pi(q), \pi(r), \pi(s), \pi(t)\}$.
Hence, $\cB$ is regular.

\end{proof}

\begin{proposition} \label{proposition.tree.hom}
A tree-shift accepted by a regular edge tree automaton is conjugate to a Hom tree-shift.
\end{proposition}
\begin{proof}
Let $X$ be the tree-shift accepted by a regular edge tree automaton $\cB= (Q, \Delta)$ with an adjacency matrix $M$. All non-null values 
of the row $p$ of $M$ are equal to $n_p$. 
We split each state $p$ of $\cB$ into $n_p$ states $p_1, \ldots p_{n_p}$ in a new edge tree automaton
$\cC$. If there is a transition $p \xrightarrow{n_p} (q, r)$ in $\Delta$, there are transitions 
$p_i \to (q_j, r_k)$ for all $1 \leq i \leq n_p$, $1 \leq j \leq n_q$, and $1 \leq k \leq n_r$. 

The automaton $\cC$ is thus an edge tree automaton (where all transitions have distinct labels) that has exactly one transition from 
a state $p_i$ to the pair $(q_j, r_k)$ for each $1 \leq i \leq n_p$, $1 \leq j \leq n_q$, and $1 \leq k \leq n_r$. Let $N$ be the adjacency matrix of~$\cC$.
Note that if $N_{p_i, (q_j, r_k)} = 1$ for some $1 \leq i \leq n_p$, $1 \leq j \leq n_q$, and $1 \leq k \leq n_r$, then  $N_{p_{i'}, (q_{j'}, r_{k'})} = 1$ for all $1 \leq i' \leq n_p$, $1 \leq j' \leq n_q$ and $1 \leq k' \leq n_r$. Note also that if $N_{p_i,(q_j, r_k)} = 1$, then $N_{p_i, (r_k, q_j)} = 1$. Indeed, if $N_{p_i, (q_j, r_k)} = 1$, then $M_{p, (q, r)} = M_{p, (r, q)} = n_p > 0$, implying $N_{p_i, (r_k, q_j)} = 1$.
   
Let us show that there is a simple undirected graph $G$ such that $\cC = \cA(G)$.
Let $G = (V, E)$, where $V$ is the set of states of $\cC$, be the undirected graph such that there is an edge $(p_i, q_j)$ in $E$ if
there is a transition $p_i \to (q_j, r_k)$ in~$\cC$ for some $r_k \in V$.
This definition is consistent with the fact that the edge $(p_i, q_j)$ is undirected. Indeed, if $N_{p_i, (q_j, r_k)} = 1$, then $M_{p, (q, r)} = n_p \neq 0$.
Thus, $M_{q, (p, p)} = n_q \neq 0$, 
implying $N_{q_j, (p_i, p_i)} = 1$ and $(q_j, p_i)$ in $E$.

It is easy to check that $\cC = \cA(G)$. As a consequence, 
$X$ is conjugate to $X_{\cA(G)}$, itself conjugate to the Hom tree-shift $X(G)$.

\end{proof}

\begin{proposition} \label{proposition.tree.conjugacy2}
It is decidable whether an edge tree shift is conjugate to a Hom tree-shift.
\end{proposition}
\begin{proof}
Let $X$ be an edge tree-shift. Without loss of generality, we assume that it is defined by a trim edge tree automaton $\cA$.
Let us show that $X$ is conjugate to a Hom tree shift if and only if the total amalgamation of $\cA$ is 
regular.

If $X$ is conjugate to a Hom tree-shift $Y$ defined a trim automaton $\cB = \cA(G)$, by Proposition \ref{proposition.tree.williams2},
$\cA$ and $\cB$ have the same total amalgamation $\cD$, up to a renaming of the states. 
By Proposition \ref{proposition.tree.regular}, this total amalgamation is regular.
Conversely, if the total amalgamation of $\cA$ is regular, by Proposition \ref{proposition.tree.hom}, this amalgamation defines an edge tree-shift conjugate
to a Hom shift.

The time and space complexity for deciding the property is polynomial in the size of $\cA$.

\end{proof}

\begin{proposition} \label{proposition.tree.isomorphic}
Two Hom tree shifts defined by the edge tree automata $\cA(G)$ and $\cA(H)$, respectively, 
are conjugate if and only if $G$ and $H$ are isomorphic as undirected graphs.
\end{proposition}
\begin{proof} 
Let $\cA(G)$ and $\cA(H)$ be the edge tree automata corresponding to $G$ and $H$, respectively, and let $X(G)$, $X(H)$ be the corresponding
tree-shifts. The tree automata $\cA(G)$ and $\cA(H)$ are regular and trim.
By Proposition \ref{proposition.tree.conjugacy2}, if $X(G)$ and $X(H)$ are conjugate, there is a regular tree automaton $\cC$ such that
$\cC$ is the total amalgamation of $\cA(G)$ and is the total amalgamation of $\cA(H)$.

Let $M$ be the adjacency matrix of $\cA(G)$. 
Let us perform an initial general amalgamation $\cB$ of $\cA(G)$, obtained using a partition $(V_1, \ldots, V_\ell)$ of the states of $\cA(G)$. 
This partition is defined according to the equivalence relation $p \sim q$ if and only if $M_{s, (p, r)} = M_{s, (q, r)}$ for all states $r, s$, and $M_{s, (r, p)} = M_{s, (r, q)}$ for all states $r, s$.

We assume that all states $p_{i,1} , \ldots p_{i, i_{n_i}}$ of $V_i$ are merged into a state $p_{i}$ in~$\cB$.
Since $\cA(G)$ is regular and since 
$M$ has $0{-}1$ entries, if $M_{p, (q, r)} = 1$ for some $p \in V_i$, $q \in V_j$, $r \in V_k$ then 
$M_{p, (q', r')} = 1$ for any $q' \in V_j$, $r' \in V_k$. This implies that
$M_{q', (p, p)} = 1$ and $M_{r', (p, p)} = 1$ for any $q' \in V_j$, $r' \in V_k$, and 
$M_{q', (p', p')} = 1$, $M_{r', (p', p')} = 1$ for any $p' \in V_i$, $q \in V_j$, $r \in V_k$.
Hence, $M_{p', (q', r')} = 1$ for any $p' \in V_i$, $q' \in V_j$, $r' \in V_k$. 
As shown in the proof of Proposition~\ref{proposition.tree.regular},
the adjacency matrix $N$ of the amalgamation \( \cB \) satisfies the following:
either $N_{p_i, (p_j, p_k)} = 0$ or $N_{p_i, (p_j, p_k)}= n_i$, where $n_i = |V_i|$.

Assume that the two (distinct) states $p_i$ and $p_j$ can be merged through another amalgamation after this round.
This implies that $N_{p_k, (p_i, p_\ell)} = N_{p_k, (p_j, p_\ell)}$ for any states $p_k, p_\ell$ in $\cB$.
Hence, $M_{p, (q, r)} = M_{p', (q', r')}$ for some $p, p' \in V_k$, $q \in V_i, q' \in V_j$, $r,r' \in V_\ell$,
implying  $M_{p, (q, r)} = M_{p, (q', r')}$ for any $p$, $q \in V_i$, $q' \in V_j$, $r, r' \in V_\ell$.
This implies $V_i = V_j$ by definition of the partition $(V_1, \ldots, V_\ell)$.   

It follows that $\cB$ is the total amalgamation $\cC$ of $\cA(G)$. 

Similarly, we consider the partition $(W_1, \ldots, W_{\ell'})$ of the states of $\cA(H)$ defined according to the equivalence relation $p \sim q$ if and only 
if $P_{s, (p, r)} = P_{s, (q, r)}$ for all states $r, s$,
 and $P_{s, (r, p)} = P_{s, (r, q)}$ for all states $r, s$, where $P$ is the adjacency matrix of $\cA(H)$.
After renumbering the sets $W_i$, we may assume that all states of $W_i$ are amalgamated to $p_i$ into the tree automaton $\cC$.
Thus, $\ell = \ell'$ and either $N_{p_i, (p_j, p_k)} = 0$ or $N_{p_i, (p_j, p_k)} = |V_i| = |W_i|$, implying
$|V_i|= |W_i|$ for each $i$. 
There is a transition from each state in $V_i$ to each pair $(q, r)$ with $q \in V_j, r \in V_k$ in $\cA(G)$ if and only if there is a transition from each state in $W_i$  to each pair $(q, r)$ with $q \in W_j, r \in W_k$ in $\cA(H)$.  Hence, $G = H$, up to a renaming of the states.
\end{proof}

\section{Directed Hom tree shifts}  \label{section.directed.hom.tree.shifts}

In this section, we consider the class of tree shifts defined by a simple (unlabeled) directed graph $G$ (where self-loops are allowed), referred to as directed Hom tree-shifts. These shifts are called hom tree-SFT in \cite{BanEtAl2021}.

Let $G=(V, E)$ be a directed graph. The \emph{directed Hom tree shift} defined by $G$, denoted $X(G)$, is the set of trees $t$ labeled on $V$ such that for each node $u$ and each $i \in \{0, 1 \ldots, d-1\}$, $(t_u, t_{ui})$ is a (directed) edge of $G$. This notion is thus weaker than the notion of (undirected) Hom tree-shift, as defined
in Section \ref{section.hom.tree.shifts}.
Without loss of generality, we may consider below only trees of arity $d = 2$.

A directed Hom tree-shift is a tree-shift of finite type. 
It is accepted by the tree automaton $\cB = (V, \Delta)$, whose transitions are $(p, p) \to (q, r)$, where $(p, q), (p, r) \in E$. 

Let $\cA(G)$ be the edge tree automaton obtained from $\cB$ by labeling all transitions with distinct labels. The tree-shift $X(G)$ is conjugate to the edge tree-shift accepted by $\cA(G)$.    
   
We say that an edge tree automaton $\cA=(Q, \Delta)$ with adjacency matrix $M$ is \emph{symmetric} if for all states $p, q, r$:
\begin{itemize}
\item $M_{p, (q, r)}  = M_{p, (r, q)}$.
\end{itemize}

Note that the adjacency matrix of $\cA(G)$ is symmetric.

\begin{proposition} \label{proposition.symmetric}
The amalgamation of a symmetric edge tree automaton is a symmetric edge tree automaton.
\end{proposition} 
\begin{proof}
Assume that two states $p_i, p_j$ of a symmetric edge tree automaton $\cA$ are merged into a state $p_{ij}$ in an edge tree automaton $\cB$. Let $M, N$ be the adjacency matrices of $\cA$ and $\cB$, respectively.

Thus, the columns of indices $(p_i, p)$ and $(p_j, p)$ (respectively $(p, p_i)$ and $(p, p_j)$) of the adjacency matrix $M$ of $\cA$ are identical for each state $p$. For a state $p$ of $\cA$, let $\pi(p)$ denote the state $p$ if $p \neq p_i, p_j$, and $\pi(p_i) = \pi(p_j) = p_{ij}$. 

Let us show that $N_{\pi(p), (\pi(q), \pi(r))} = N_{\pi(p), (\pi(r), \pi(q))}$.

If $N_{\pi(p), (\pi(q), \pi(r))} > 0$, then there exist $p'\in V$ with $\pi(p') = \pi(p)$, such that for all $q', r' \in V$ with $\pi(q') = \pi(q)$, $\pi(r') = \pi(r)$, $M_{p', (q', r')}  > 0$. This implies that $M_{p', (r', q')} > 0$ and $M_{p', (q', r')} = M_{p', (r', q')}$, implying
$N_{\pi(p), (\pi(r), \pi(q))} > 0$. 

If $\pi(p) \neq p_{ij}$, then $N_{\pi(p), (\pi(q), \pi(r))} = M_{p, (q', r')}  = M_{p, (r', q')} = N_{\pi(p), (\pi(r), \pi(q))}$ for all $q', r' \in V$ with $\pi(q') = \pi(q)$, $\pi(r') = \pi(r)$.

If $\pi(p) =  p_{ij}$, then $p = q_i$ or $p = q_j$, and $p' = q_i$ or $p' = q_j$.
Thus, for all $q', r' \in V$ with $\pi(q') = \pi(q)$, $\pi(r') = \pi(r)$, $N_{\pi(p), (\pi(q), \pi(r))} = M_{p_i, (q', r')} + M_{p_j, (q', r')} = M_{p_i, (r', q')} + M_{p_j, (r', q')} = N_{\pi(p), (\pi(r), \pi(q))}$.

Hence, $\cB$ is symmetric.
\end{proof}

\begin{proposition} \label{proposition.symmetric2}
A tree-shift accepted by a symmetric edge tree automaton is conjugate to a directed Hom tree-shift.
\end{proposition}
\begin{proof}
Let $X$ be the tree-shift accepted by a symmetric edge tree automaton $\cB= (Q, \Delta)$ with adjacency matrix $M$.
For all states $p, q, r$ of $\cB$ such that $M_{p, (q, r) > 0}$,
we split the state $p$ into the states $p_{\{q, r\}, 1}, \ldots p_{\{q, r\}, M_{p,(q,r)}}$ in an edge tree automaton
$\cC$. 
If $M_{p,(q,r)} > 0$, then there are in $\cC$ transitions 
$p_{\{q,r\}, i}  \to (q_{\{q', r'\}, j}, r_{\{q'', r''\}, k})$ 
and
$p_{\{q,r\}, i}  \to (r_{\{q'', r''\}, j}, q_{\{q', r'\}, k})$ 
for all $1 \leq i \leq M_{p,(q,r)}$, $1 \leq j \leq M_{q,(q',r')}$, and $1 \leq k \leq M_{r,(q",r")}$,
for all $q', r', q'', r''$ such that $M_{q,(q',r')}, M_{r,(q'',r'')} > 0$.
The automaton $\cC$ is thus an edge tree automaton (where all transitions have distinct labels) that has exactly one transition from 
the state $p_{\{q,r\}, i}$ to the pair $(q_{\{q', r'\}, j}, r_{\{q'', r''\}, k})$. Note that the definition is consistent since
$M_{p,(q,r)}  = M_{p,(r, q)}$ for all states $p, q, r$.

Let $N$ be the adjacency matrix of $\cC$.
Note that if there exist $i, j, k, q', r', q'', r''$ with $1 \leq i \leq M_{p,(q,r)}$, $1 \leq j \leq M_{q,(q',r')}$ and $1 \leq k \leq M_{r,(q'',r'')}$ such that
$N_{p_{\{q,r\}, i}, (q_{\{q', r'\}, j}, r_{\{q'', r''\}, k})} = 1$,
then $N_{p_{\{q,r\}, i}, (q_{\{q', r'\}, j}, r_{\{q'', r''\}, k})} = 1$
for all $i, j, k, q', r', q'', r''$ with $1 \leq i \leq M_{p,(q,r)}$, $1 \leq j \leq M_{q,(q',r')}$
and $1 \leq k \leq M_{r,(q'',r'')}$.

Let us show that there is a simple directed graph $G$ such that $\cC = \cA(G)$.
Let $G = (V, E)$, where $V$ is the set of states of $\cC$, be the directed graph such that there is an edge 
$(p_{\{q, r\}, i}, q_{\{u, v\}, j})$ in $E$ if
there is a transition from $p_{\{q, r\}, i}$  to $(q_{\{u, v\}, j}, r_{\{u', v'\}, k})$ in $\cC$ for some $r, u', v'  \in V$.
One can check that $\cC = \cA(G)$. As a consequence, 
 $X$ is conjugate to $X_{\cA(G)}$, itself conjugate to the directed Hom tree-shift $X(G)$.
\end{proof}

\begin{proposition} \label{proposition.directed.tree.conjugacy2}
It is decidable whether an edge tree shift is conjugate to a directed Hom tree-shift.
\end{proposition}

\begin{proof}
Let $X$ be an edge tree-shift. Without loss of generality, we assume that it is defined by a trim edge tree automaton $\cA$.
Let us show that $X$ is conjugate to a directed Hom tree shift if and only if the total amalgamation of $\cA$ is 
symmetric.

If $X$ is conjugate to a directed Hom tree-shift $Y$ defined by a trim automaton $\cB = \cA(G)$, by Proposition \ref{proposition.tree.williams2},
$\cA$ and $\cB$ have the same total amalgamation $\cD$, up to a renaming of the states. 
By Proposition \ref{proposition.symmetric}, this total amalgamation is symmetric.
Conversely, if the total amalgamation of $\cA$ is symmetric, by Proposition \ref{proposition.symmetric2}, this amalgamation defines an edge tree-shift conjugate
to a directed Hom shift.

The time and space complexity for deciding the property is polynomial in the size of $\cA$.
\end{proof}

\begin{example}
Let $\cA$ be the (non-symmetric) tree automaton with states $p, q, r, s, u$ and transition matrix:
\begin{align*}
M &= 
\begin{blockarray}{cccccccccc}
 & (q, r) & (r, q) & (p, p) & (p, s) & (s, p) & (s, s) & (r, u) & (u, r)\\
\begin{block}{c(ccccccccc)}
  p & 1 & 0 & 0 & 0 & 0 & 0 & 1 & 1\\
  q & 0 & 0 & 1 & 1 & 1 & 1 & 0  & 0\\
  r & 0 & 0 & 0 & 0 & 0 & 0 & 1 & 1\\
  s & 0 & 1 & 0 & 0 & 0 & 0 & 1 & 1\\
  u & 1 & 1 & 0 & 0 & 0 & 0 & 0 & 0\\
\end{block}
\end{blockarray}.
\end{align*}
Its total amalgamation is obtained by merging states $p$ and $s$ into a state $(ps)$:
\begin{align*}
N &= 
\begin{blockarray}{cccccc}
 & (q, r) & (r, q) & ((ps), (ps))  & (r, u) & (u, r)\\
\begin{block}{c(ccccc)}
  (ps) & 1 & 1 & 0 & 2 & 2\\
  q & 0 & 0 & 1 & 0 & 0\\
  r & 0 & 0 & 0 & 1 & 1\\
  u & 1 & 1 & 0 & 0 & 0\\
\end{block}
\end{blockarray}.
\end{align*}   
The edge tree automaton $\cB$ with adjacency matrix $N$ is symmetric. Thus, the tree-shift $X$ accepted by $\cA$
is conjugate to the directed Hom tree shift defined by the following directed graph $G$ of Figure \ref{figure.directed}.

\vspace{-5mm}
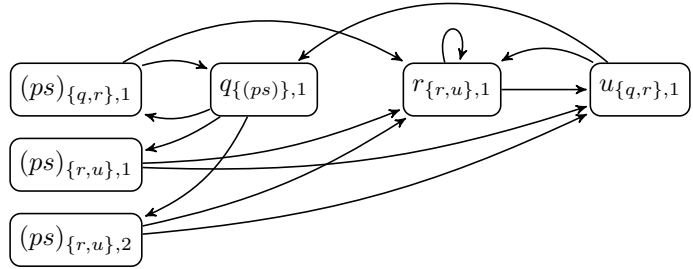
\begin{figure}[htbp]
\centering
    \begin{tikzpicture}[shorten >=1pt,node distance=2.5cm, every edge/.style={draw,->,>=stealth',auto,semithick},auto]
    \tikzset{every state/.append style={rectangle, rounded corners}}
      \node[state] (1) at (0,0) {$(ps)_{\{q,r\}, 1}$};
       \node[state] (11) at (0,-1) {$(ps)_{\{r,u\}, 1}$};
         \node[state] (111) at (0,-2) {$(ps)_{\{r,u\}, 2}$};
      \node[state, right of=1] (2) {$q_{\{(ps)\}, 1}$};
      \node[state, right of=2] (3) at (2.5,0) {$r_{\{r, u \}, 1}$};
        \node[state, right of=3] (5) {$u_{\{q, r \}, 1}$};
      \draw (1) edge[bend left=20] node {} (2);
     \draw (1) edge[bend left=30] node {} (3);
      \draw (2) edge[bend left=20] node {} (1);
       \draw (2) edge[bend left=10] node {} (11);
      \draw (2) edge[bend left=20] node {} (111);
      \draw (3) edge[loop above] node {} (3);
       \draw (3) edge node {} (5);
        \draw (5) edge[bend right=40] node {} (2);
        \draw (5) edge[bend right=30] node {} (3);
          \draw (11) edge[bend right=10] node {} (3);
           \draw (111) edge[bend right=10] node {} (3);
           \draw (11) edge[bend right=10] node {} (5);
           \draw (111) edge[bend right=10] node {} (5);
    \end{tikzpicture}
      \caption{The directed graph G defining the directed Hom tree-shift conjugate to $X$.}\label{figure.directed}
\end{figure}
\end{example}

\newpage 

\section*{Acknowledgments}
This work was supported by the Agence Nationale
de la Recherche (ANR-22-CE40-0011).
The second named author was supported by Cofund Math In Greater Paris, Marie Sk\l odowska-Curie Actions (H2020-MSCA-COFUND-2020-GA101034255).
\nocite*
\bibliographystyle{plain}
\bibliography{oneSidedHomShift}

@book{LindMarcus1995,
	author = {Lind, Douglas and Marcus, Brian},
	pages = {xvi+495},
	publisher = {Cambridge University Press, Cambridge},
	title = {An Introduction to Symbolic Dynamics and Coding},
	year = {1995},
	note={2nd edition, 2021}
	}

@book {Kitchens1998,
    AUTHOR = {Kitchens, Bruce P.},
     TITLE = {Symbolic dynamics},
    SERIES = {Universitext},
      NOTE = {One-sided, two-sided and countable state Markov shifts},
 PUBLISHER = {Springer-Verlag},
   ADDRESS = {Berlin},
      YEAR = {1998},
     PAGES = {x+252},
      ISBN = {3-540-62738-3},}

@incollection{BealBerstelEilersPerrin2010,
  author	= {B\'eal, Marie-Pierre and Berstel, Jean and Eilers, Soren and Perrin, Dominique},
  title		= {Symbolic Dynamics},
  booktitle = {{H}andbook of {A}utomata {T}heory},
  year  = {2021},
  volume = {II},
  chapter = {27},
  pages = {987-1031},
  editor = {Jean-{\'E}ric Pin},
  publisher ={EMS Press},
    archivePrefix = {arXiv},
  eprint    = {1006.1265},
}

@article {Williams1973,
    AUTHOR = {Williams, R. F.},
     TITLE = {Classification of subshifts of finite type},
     JOURNAL = {Annals of Mathematics},
      VOLUME = {98},
      NUMBER = {1},
     PAGES = {120--153},
      YEAR = {1973},
      Note = {erratum 99 (1974), 380-381}}

@article{BoyleFranksKitchens1990, 
title={Automorphisms of one-sided subshifts of finite type}, 
volume={10}, 
DOI={10.1017/S0143385700005678}, 
number={3}, 
journal={Ergodic Theory and Dynamical Systems}, 
author={Boyle, Mike and Franks, John and Kitchens, Bruce}, 
year={1990}, 
pages={421–449}}

@article{CHANDGOTIA_2017, 
title={Four-cycle free graphs, height functions, the pivot property and entropy minimality}, 
volume={37}, 
DOI={10.1017/etds.2015.88}, 
number={4}, 
journal={Ergodic Theory and Dynamical Systems}, 
author={Chandgotia, Nishant}, 
year={2017}, 
pages={1102–1132}
}

@article{ChandgotiaMarcus2018,
  title={Mixing Properties for Hom-Shifts and the Distance between Walks on Associated Graphs},
  author={Nishant Chandgotia and Brian H. Marcus},
  journal={Pacific Journal of Mathematics},
  volume={294},
  number={1},
  pages={41--69},
  year={2018}
}

@article{GangloffHellouinOpocha2024,
title = {Short-range and long-range order: a transition in block-gluing behavior in Hom shifts},
author={Silv{\`e}re Gangloff and Benjamin Hellouin de Menibus and Piotr Opocha},
  journal={Journal d’analyse math{\'e}matique},
  year = {2024},
  note ={inPress, hal-03842725v2}}

@article{AubrunBeal2012,
author={Nathalie Aubrun and Marie-Pierre Béal},
title={Tree-shifts of finite type},
journal={Theoretical Computer Science},
Volume={459},
pages={16-25},
year={2012}
}

@article{AubrunBeal2013,
  author    = {Nathalie Aubrun and
               Marie{-}Pierre B{\'{e}}al},
  title     = {Sofic Tree-Shifts},
   JOURNAL = {Theory Comput. Syst.},
  FJOURNAL = {Theory of Computing Systems},  
  number    = {4},
  pages     = {621--644},
  year      = {2013},
  url       = {http://dx.doi.org/10.1007/s00224-013-9456-1},
}

@article{BanCHW22,
  author       = {Jung{-}Chao Ban and
                  Chih{-}Hung Chang and
                  Wen{-}Guei Hu and
                  Yu{-}Liang Wu},
  title        = {Topological entropy for shifts of finite type over {Z} and trees},
  journal      = {Theor. Comput. Sci.},
  volume       = {930},
  pages        = {24--32},
  year         = {2022},
  url          = {https://doi.org/10.1016/j.tcs.2022.07.007},
  doi          = {10.1016/J.TCS.2022.07.007},
}

@article {BanEtal2025a,
    AUTHOR = {Ban, Jung-Chao and Chang, Chih-Hung and Huang, Nai-Zhu and
              Lai, Guan-Yu},
     TITLE = {Transitivity and dense orbits for {M}arkov tree-shifts},
   JOURNAL = {Nonlinearity},
  FJOURNAL = {Nonlinearity},
    VOLUME = {38},
      YEAR = {2025},
    NUMBER = {7},
     PAGES = {Paper No. 075021, 16},
      ISSN = {0951-7715,1361-6544},
   MRCLASS = {37B51 (37B05 37B20)},
  MRNUMBER = {4921394},
       DOI = {10.1088/1361-6544/addcb0},
       URL = {https://doi.org/10.1088/1361-6544/addcb0},
}

@article {BanEtAl2025b,
    AUTHOR = {Ban, Jung-Chao and Lai, Guan-Yu and Wu, Yu-Liang},
     TITLE = {Hausdorff dimensions of irreducible {M}arkov hom tree-shifts},
   JOURNAL = {J. Lond. Math. Soc. (2)},
  FJOURNAL = {Journal of the London Mathematical Society. Second Series},
    VOLUME = {111},
      YEAR = {2025},
    NUMBER = {6},
     PAGES = {Paper No. e70198, 39},
      ISSN = {0024-6107,1469-7750},
   MRCLASS = {60J10 (28A80)},
  MRNUMBER = {4920143},
       DOI = {10.1112/jlms.70198},
       URL = {https://doi.org/10.1112/jlms.70198},
}

@article{BanEtAl2021,
title = {Characterization and topological behavior of homomorphism tree-shifts},
journal = {Topology and its Applications},
volume = {302},
pages = {107848},
year = {2021},
issn = {0166-8641},
doi = {https://doi.org/10.1016/j.topol.2021.107848},
url = {https://www.sciencedirect.com/science/article/pii/S0166864121002637},
author = {Jung-Chao Ban and Chih-Hung Chang and Wen-Guei Hu and Guan-Yu Lai and Yu-Liang Wu},
}

@article {PetersenSalama2020,
    AUTHOR = {Petersen, Karl and Salama, Ibrahim},
     TITLE = {Entropy on regular trees},
   JOURNAL = {Discrete Contin. Dyn. Syst.},
  FJOURNAL = {Discrete and Continuous Dynamical Systems. Series A},
    VOLUME = {40},
      YEAR = {2020},
    NUMBER = {7},
     PAGES = {4453--4477},
      ISSN = {1078-0947,1553-5231},
   MRCLASS = {37B10 (37B40 37E25)},
  MRNUMBER = {4097549},
MRREVIEWER = {Andriy\ S.\ Ol\=\i\u inik},
       DOI = {10.3934/dcds.2020186},
       URL = {https://doi.org/10.3934/dcds.2020186},
}

@article {PetersenSalama2018,
    AUTHOR = {Petersen, Karl and Salama, Ibrahim},
     TITLE = {Tree shift topological entropy},
   JOURNAL = {Theoret. Comput. Sci.},
  FJOURNAL = {Theoretical Computer Science},
    VOLUME = {743},
      YEAR = {2018},
     PAGES = {64--71},
      ISSN = {0304-3975,1879-2294},
   MRCLASS = {37B40 (37B10 37E25)},
  MRNUMBER = {3842480},
       DOI = {10.1016/j.tcs.2018.05.034},
       URL = {https://doi.org/10.1016/j.tcs.2018.05.034},
}

@article {PetersenSalama2023,
    AUTHOR = {Petersen, Karl and Salama, Ibrahim},
     TITLE = {Asymptotic pressure on some self-similar trees},
   JOURNAL = {Stoch. Dyn.},
  FJOURNAL = {Stochastics and Dynamics},
    VOLUME = {23},
      YEAR = {2023},
    NUMBER = {2},
     PAGES = {Paper No. 2350009, 19},
      ISSN = {0219-4937,1793-6799},
   MRCLASS = {37B10 (37B40 37B51 82B20)},
  MRNUMBER = {4577495},
       DOI = {10.1142/S0219493723500090},
       URL = {https://doi.org/10.1142/S0219493723500090},
}

@book{Tata2021,
	author = {Hubert Comon
              and Max Dauchet
              and Rémi Gilleron
              and Florent Jacquemard
              and Denis Lugiez
              and Christof Löding
              and Sophie Tison
              and  Marc Tommasi},
	pages = {262},
	title = {Tree Automata Techniques and Applications},
	year = {2021},
   archivePrefix = {HAL},
    URL = {https://inria.hal.science/hal-03367725v1},
note = {hal-03367725},
	}
\end{document}